\documentclass[11pt]{article}
\usepackage{fullpage}
\usepackage{vmargin}
\usepackage{comment}
\usepackage{enumerate}
 \setmarginsrb{1in}{1in}{1in}{1in}{0pt}{0pt}{0pt}{7mm}
\usepackage{tikz}
\usetikzlibrary{snakes}

\usepackage{makeidx}  

\usepackage{amsmath,amstext,amssymb}
\usepackage{xspace}
  \usepackage{amsthm}

  \newtheorem{theorem}{Theorem}
  \newtheorem{lemma}{Lemma}
    \newtheorem{corollary}{Corollary}
  
  \newtheorem{definition}{Definition}
  
  \newtheorem{observation}{Observation}
  \newtheorem{claim}{Claim}

\begin{document}
%

\newcommand{\tudu}[1]{{\bf{TODO: #1}}}
\newcommand{\transname}[2]{$#1 \slash #2$-{\sc{transversal}}}
\newcommand{\himmersion}{{$\mathbf{H}$\sc{-Immersion}}\xspace}
\newcommand{\dpath}{{\sc{Edge Disjoint Paths}}\xspace}
\newcommand{\toursection}{{\sc{Semi-complete Digraph Balanced Cut}}\xspace}
\newcommand{\vertexpath}{{\sc{Vertex Disjoint Paths}}\xspace}
\newcommand{\SAT}{{\sc{SAT}}\xspace}
\newcommand{\tcnfsat}{{\sc{3-CNF-SAT}}\xspace}
\newcommand{\compass}{\ensuremath{\textrm{NP} \subseteq \textrm{coNP}/\textrm{poly}}}
\newcommand{\cw}{\mathbf{cw}}
\newcommand{\ctw}{\mathbf{ctw}}
\newcommand{\pw}{\mathbf{pw}}
\newcommand{\N}{\ensuremath{\mathbb{N}}}
\newcommand{\Aa}{\ensuremath{\mathcal{A}}}
\newcommand{\Bb}{\ensuremath{\mathcal{B}}}
\newcommand{\Jj}{\ensuremath{\mathcal{J}}}
\newcommand{\Pp}{\ensuremath{\mathcal{P}}}
\newcommand{\Qq}{\ensuremath{\mathcal{Q}}}
\newcommand{\Ss}{\ensuremath{\mathcal{S}}}
\newcommand{\Ii}{\ensuremath{\mathcal{I}}}
\newcommand{\randfamily}{\ensuremath{\mathcal{F}}}
\newcommand{\app}{$\star$}

\newcommand{\WO}{{\sc{W[1]}}\xspace}
\newcommand{\FPT}{{\sc{FPT}}\xspace}
\newcommand{\NP}{{\sc{NP}}\xspace}
\newcommand{\XP}{{\sc{XP}}\xspace}
\newcommand{\MSO}{\ensuremath{\mathbf{MSO}}}
\newcommand{\ff}{\ensuremath{\varphi}}

\newcommand{\FF}{\ensuremath{\mathfrak{F}}}
\newcommand{\Ff}{\ensuremath{\mathcal{F}}}
\newcommand{\Gg}{\ensuremath{\mathcal{G}}}
\newcommand{\UU}{\ensuremath{\mathfrak{U}}}
\newcommand{\ptrace}{trace}
\newcommand{\ptraces}{traces}

\newcommand{\vf}[2]{\mathcal{V}_{(#1,#2)}}
\newcommand{\ef}[2]{\mathcal{E}_{(#1,#2)}}
\newcommand{\nequiv}{\not\equiv}

\newcommand{\defproblemu}[3]{
  \vspace{1mm}
\noindent\fbox{
  \begin{minipage}{0.95\textwidth}
  #1 \\
  {\bf{Input:}} #2  \\
  {\bf{Question:}} #3
  \end{minipage}
  }
  \vspace{1mm}
}
\newcommand{\defparproblem}[4]{
  \vspace{1mm}
\noindent\fbox{
  \begin{minipage}{0.97\textwidth}
  \begin{tabular*}{\textwidth}{@{\extracolsep{\fill}}lr} #1  & {\bf{Parameter:}} #3 \\ \end{tabular*}
  {\bf{Input:}} #2  \\
  {\bf{Question:}} #4
  \end{minipage}
  }
  \vspace{1mm}
}

\title{Jungles, bundles, and fixed parameter tractability}

\author{Fedor V. Fomin
\thanks{Department of Informatics, University of Bergen,  Bergen, Norway,  \{fomin,michal.pilipczuk\}@ii.uib.no} \and  Micha\l{} Pilipczuk~$^*$
}

\begin{titlepage}
\def\thepage{}
\thispagestyle{empty}
\maketitle

\begin{abstract}
We give a fixed-parameter tractable (\FPT) approximation algorithm computing the pathwidth of a tournament, and more generally, of a semi-complete digraph.
Based on this result, we prove the following.
\begin{itemize}
\item The {\sc{Topological Containment}} problem is \FPT on semi-complete digraphs. More precisely, given a semi-complete $n$-vertex digraph $T$ and a digraph $H$, one can check in time  $f(|H|)n^3\log n$, where $f$ is some elementary function, whether $T$ contains a subdivision of $H$ as a subgraph. The previous known algorithm for this problem was due to Fradkin and Seymour and was of running time $n^{m(|H|)}$, where $m$ is a quadruple-exponential function.
\item The {\sc{Rooted Immersion}} problem is \FPT on semi-complete digraphs. The complexity of this problem was left open by Fradkin and Seymour. Our algorithm solves it in time $g(|H|)n^4\log n$, for some elementary function $g$.
\item Vertex deletion distance to every immersion-closed class of semi-complete digraphs is fixed-parameter tractable. More precisely, for every immersion-closed class $\Pi$ of semi-complete digraphs, there exists an algorithm with running time $h(k)n^3\log n$ that checks, whether one can remove at most $k$ vertices from a semi-complete $n$-vertex digraph to obtain a digraph from class $\Pi$. Here, $h$ is some function depending on the class~$\Pi$.
\end{itemize}


\end{abstract}



\end{titlepage}

\section{Introduction}\label{sec:introduction}

\medskip
\noindent\textbf{Graph Minors Theory} of Robertson and Seymour~\cite{RobertsonS_GM} significantly reshaped the modern way of thinking about graphs. Apart from being a beautiful mathematical concept, the theory offers a number of fundamental algorithmic applications~\cite{RS95,RobertsonS2}. Arguably, the most fruitful of these is based on obstructions to admitting tree decompositions of small width. The structure of obstructions allows one to use the following WIN/WIN approach. If the treewidth of the input graph turns out to be small, it is possible to solve the problem efficiently, e.g., by applying dynamic programming. However, if the treewidth is large, we can find a sufficiently large obstruction, whose structure can be used either to immediately provide an answer, or to identify a vertex or an edge that is irrelevant to the solution, and hence can be safely removed from the graph. This approach is the corner stone of Robertson and Seymour seminal  fixed-parameter tractable (\FPT) algorithm checking, whether the input (undirected) graph $G$ contains some fixed graph $H$ as a minor~\cite{RS95}. As a consequence, various graph problems have been shown to have polynomial-time algorithms, some of which were previously not even known to be decidable~\cite{FellowsL88}.

In combination with algorithms for computing the treewidth of a graph \cite{Amir10,Bodlaender96} and powerful logic tools formalizing graph properties \cite{ArnborgLS91,Courcelle90,Courcelle92a}, the WIN/WIN paradigm became ubiquitous in (undirected) graph algorithms. Applications of this approach include the algorithm of Robertson and Seymour for the disjoint paths problem \cite{RS95}, decidability of Hadwiger's conjecture \cite{KawarabayashiR09}, model-checking \cite{DawarGK07},  bidimensionality \cite{DFHT05,FominLRS10,F.V.Fomin:2010oq}, graph embeddings \cite{Grohe01,KawarabayashiMR08}, minor ideals \cite{AGK08}, and topological minor containment \cite{GroheKMW11} among many others.

\medskip
\noindent\textbf{Directed graphs.} It is a very natural and important question, whether techniques and results of Graph Minors can be applied in the world of \emph{directed} graphs or digraphs. In spite of many attempts, we are still very far from the right answer. Even to capture a good analogue for treewidth in digraphs is a non-trivial task and several notions like directed treewidth~\cite{JohnsonRST01-Di}, DAG-width~\cite{BerwangerDHK06,Obdrzalek06-DA} or Kelly-width~\cite{KreutzerH07-Di} can be found in the literature. However, none of them shares all the ``nice" properties of undirected treewidth. In fact this claim can be formalized and proved; Ganian et al. argued that ``{any reasonable} algorithmically useful and structurally nice digraph measure cannot be substantially different from the treewidth of the underlying undirected graph" \cite{GanianHKMORS10}.

The notion of graph minor is crucial in defining obstructions to small treewidth in an undirected graph. There are several ways to generalize this definition to digraphs and, as in case of treewidth, it is unclear which of them is the most natural. One approach is to consider topological embeddings or immersions. A graph $H$ is a \emph{topological subgraph} (or \emph{topological minor}) of graph $G$ if $H$ can be obtained from a subgraph of $G$ by subdividing some of its edges. In other words, graph $H$ can be embedded into graph $G$ in such a way that edges of $H$ correspond to vertex-disjoint paths in $G$. An \emph{immersion} of a graph $H$ into a graph $G$ is defined like a topological embedding, except that edges of $H$ correspond to edge-disjoint paths in $G$. Long standing open questions were to decide if a graph $H$ can be topologically embedded (immersed) into $G$ is fixed-parameter tractable, when parameterized by the size of $H$. Both questions were answered positively only very recently by Grohe, Kawarabayashi, Marx, and Wollan~\cite{GroheKMW11}. Unfortunately, the work of Grohe et al. cannot be extended to directed graphs. By the classical result of Fortune, Hopcroft, and  Wyllie~\cite{FortuneHW80} the problem of testing whether a given digraph $G$ contains $H$ as a (directed) topological subgraph is \NP-complete even for very simple digraphs $H$ of constant size. Similar results can be easily obtained for immersions. Therefore, the most important two special cases of \emph{rooted} topological embedding and immersions, namely {\sc{Vertex (Edge) Disjoint Paths}} problems, were studied intensively on different classes of directed graphs. For example if we constrain the input digraphs to be acyclic, then both variants still remain \NP-complete~\cite{EvenIS76}, but are polynomial-time solvable for every constant number of terminal pairs $k$~\cite{FortuneHW80}, which was not the case in the general setting~\cite{FortuneHW80}. Slivkins showed that both problems are in fact \WO-hard when parameterized by $k$~\cite{Slivkins10}, thus completing the picture of their parameterized complexity in this restricted case.

\medskip
\noindent\textbf{Tournaments} form an interesting and mathematically rich subclass of digraphs. Many algorithmic problems were studied on tournaments, including problems strongly related to our work: {\sc{Vertex (Edge) Disjoint Paths}} and {\sc{Feedback Arc (Vertex) Set}} problems. In particular, {\sc{Vertex (Edge) Disjoint Paths}} problems remain NP-complete on tournaments when $k$ is a part of the input.
However, they  were known to be solvable in polynomial time for the case $k=2$ on a superclass of tournaments called semi-complete graphs \cite{Bang-Jensen91,Bang-JensenT92}. Different algorithmic aspects of the {\sc{Feedback Arc (Vertex) Set}} problem are discussed in~\cite{AilonCN08,Alon06,AlonLS09,BessyFGPPST09,CharbitTHY07,Moon:1971xz,RamanS06}.

Recent work of Chudnovsky, Fradkin, Scott, and  Seymour \cite{ChudnovskyFS2011,ChudnovskySS2011,ChudnovskyS11,fradkin-seymourEDP,fradkin-seymour} drastically advanced the study of minor-related problems in tournaments. One of the main tools in this novel approach is the cutwidth of a digraph, which in tournament setting plays a similar role to the treewidth of undirected graphs. The main algorithmic advantage of cutwidth is that there is an \FPT-approximation algorithm that in \FPT time either finds a vertex ordering of width at most $f(k)$, or constructs an obstruction for cutwidth $k$, for some function $f$~\cite{fradkin-seymourEDP}. By making use of cutwidth, Fradkin and Seymour gave an \FPT algorithm for the {\sc{Edge Disjoint Paths}} problem in tournaments. Similarly, for every digraph $H$ of size $k$ and an $n$-vertex tournament $T$, deciding if $H$ has immersion in $T$ can be done in time $f(k)n^{4}$, for some function $f$~\cite{fradkin-seymourEDP}.

The main limitation of cutwidth is that this parameter is not robust subject to even small changes in the digraph. For example, the cutwidth of an acyclic tournament with $n$ vertices is zero, but adding just one vertex to this tournament can increase its cutwidth up to $\lfloor n/2 \rfloor$. In order to circumvent this limitation, Fradkin and Seymour in \cite{fradkin-seymour} obtained an obstruction theorem for pathwidth of a tournament, by showing that tournaments of large pathwidth always contain structures called jungles, and more refined objects called triples. 
Fradkin and Seymour also gave an  algorithm that either computes a path decomposition of a tournament of width at most $f(k)$, or outputs a triple of size $k$, for some function $f$.
The running time of their algorithm is $n^{m(k)}$ for  some function $m$; algorithms with such running times are often referred as \XP algorithms, see  the book of Flum and Grohe \cite{FlumGrohebook} for the formal definition of class \XP. 
 Pathwidth is a more general width parameter than cutwidth, so the Fradkin-Seymour obstruction theorem can be used to tackle more general problems. For example,
this brings Fradkin and Seymour to an \XP algorithm that for every digraph $H$ of size $k$ and $n$-vertex semi-complete digraph $T$, decides if $H$ is topologically contained in $T$ in time $n^{m(k)}$, which is polynomial for fixed $k$~\cite{fradkin-seymour}. The only obstacle for using pathwidth as a tool for obtaining \FPT algorithms on tournaments was that there were no algorithms known constructing or approximating the pathwidth of a tournament in \FPT time.


\medskip
\noindent\textbf{Our results.} We eliminate the existing limitations of pathwidth and show that it is possible in FPT time   either to construct a path decomposition of ``small" width, or to find a ``large" triple, which topologically contains all ``small" digraphs and is an obstruction for pathwidth $k$. Our algorithm is based on structural theorems of Fradkin and Seymour \cite{fradkin-seymour} but requires several new ideas.
More precisely, we prove the following theorem. (In Section~\ref{sec:prelim} we define semi-complete digraphs --- a superclass of tournaments, pathwidth, jungles and triples)

\begin{theorem}\label{thm:bundle}
There is an algorithm that for a given
semi-complete graph $T=(V,E)$ and an integer $k$, in time
$2^{O(k\log{k})}|V|^3 \log{|V|}$ outputs
\begin{itemize}
\item[(1)] either a path decomposition $W$ of width at most $4k^2+7k$,
\item[(2)] or a $k$-jungle $\Jj$.
\end{itemize}
\end{theorem}

Fradkin and Seymour prove that if a semi-complete digraph contains a $f(k)$-jungle, then it contains also a $k$-triple, in which every digraph of size $k$ can be topologically embedded. From a close examination of their proof one can extract an algorithm with running time $O(|V|^3 \log |V|)$. Thus, if the digraph has large pathwidth, then every small enough digraph can be topologically embedded into it. On the other hand, it easy to show that the cliquewidth of a semi-complete digraph is always at most its pathwidth plus $2$. Since topological minor containment problem is expressible in $\MSO_1$, the results of Courcelle, Makowsky and Rotics \cite{CourcelleMR00} immediately give an \FPT algorithm parameterized by the width of the obtained path decomposition and the size of the embedded digraph. This implies that \textsc{Topological containment} (we give precise definition of the parameterized problem \textsc{Topological containment} in Section~\ref{sec:prelim}) in semi-complete digraphs is \FPT, which substantially improves the algorithm of Fradkin and Seymour~\cite{fradkin-seymour}. We remark that it is also possible to construct an explicit dynamic programming algorithm working on path decomposition, thus circumventing the usage of tools from logic.
\begin{theorem}\label{thm:main}
There exists an algorithm, which given a digraph $H$ and a semi-complete digraph $T=(V,E)$, checks whether $H$ is topologically contained in $T$ in time $g(|H|)|V|^3 \log{|V|}$ for some elementary function $g$.
\end{theorem}

Next application of Theorem~\ref{thm:bundle} shows that the \emph{rooted} immersion, which is a common generalization of immersion and the Edge Disjoint Paths problem, is \FPT on semi-complete digraphs.
 \begin{theorem}\label{thm:immersion}
There exists an algorithm, which given a rooted digraph $(H;v_1,v_2,\ldots,v_t)$ and a semi-complete rooted digraph $(T;w_1,w_2,\ldots,w_t)$, checks whether $H$ can be immersed into $T$ while preserving roots in time $g(|H|)|V(T)|^4\log |V(T)|$ for some elementary function $g$.
\end{theorem}
The complexity of this problem was left open by Fradkin and Seymour in~\cite{fradkin-seymourEDP}, where they showed that its close relative, rooted infusion, is \FPT on semi-complete digraphs. As rooted infusion generalizes the {\sc{Edge Disjoint Paths}} problem as well, the complexity of $k$-\textsc{Edge Disjoint Paths} was already established by Fradkin and Seymour; however, the approach presented in~\cite{fradkin-seymourEDP} is much more complicated than ours.
Our proof is based on a non-trivial application of the irrelevant vertex technique, which allows to identify a vertex in semi-complete digraphs of large pathwidth, whose removal does change the answer to the problem. It is possible to use the ideas of Ganian et al.~\cite{GanianHO11} for the {\textsc{Edge Disjoint Paths}} problem to prove that \textsc{Rooted Immersion} is expressible in $\MSO_1$, which in combination with the irrelevant vertex procedure and Theorem~\ref{thm:bundle} shows that {Rooted Immersion} is \FPT on semi-complete digraphs. Unfortunately, the formula obtained in this way has unbounded quantifier alternation, so a direct application of the framework of Courcelle, Makowsky and Rotics~\cite{CourcelleMR00} gives a non-elementary dependence on $|H|$, i.e., the tower of height depending on $|H|$ in the exponent of the bound on the running time. To overcome this we give an explicit dynamic programming algorithm.

Finally, we observe that our framework can be used to prove meta-theorems of more general nature. Let $\Pi$ be a class of digraphs and denote by $\Pi+kv$ the class of digraphs, from which one can delete at most $k$ vertices to obtain a member of $\Pi$. We study the following problem:

\defparproblem{ \textsc{$\Pi +kv$ Recognition}}{Digraph $D$ and integer $k$}{$k$}{Is there $S\subseteq V(D)$, $|S|\leq k$, such that $D\setminus S\in \Pi$?}

We are interested in classes $\Pi$ which are closed under immersion. For example the class of acyclic digraphs, or digraphs having cutwidth at most $c$, where $c$ is some constant, are of this type. In particular, the parameterized {\sc{Feedback Vertex Set}} in directed graphs is equivalent to {\sc{$\Pi +kv$ Recognition}} for $\Pi$ being the class of acyclic digraphs. Chudnovsky and Seymour~\cite{ChudnovskyS11} showed that immersion order on semi-complete digraphs is a well-quasi-order. Based on this result and Theorem~\ref{thm:bundle}, we are able to prove the following meta-theorem. Note that it seems difficult to obtain the results of this flavor using cutwidth, as cutwidth can decrease dramatically even when one vertex is deleted from the digraph.

\begin{theorem}\label{thm:immersionclosed}
Let $\Pi$ be an immersion-closed class of semi-complete digraphs. Then {\sc{$\Pi +kv$ Recognition}} is \FPT on semi-complete digraphs.
\end{theorem}

To complete the picture of the presented framework, we investigate the possibilities of proving logical meta-theorems yielding \FPT algorithms working on path decomposition of semi-complete digraphs. As pathwidth resembles treewidth in the undirected case, one could conjecture that, similarly to Courcelle's Lemma, $\MSO_2$ model checking should be also fixed-parameter tractable, when the pathwidth of the given semi-complete digraph and length of the formula are parameters. The following result shows that this is unfortunately not true.

\begin{theorem}\label{thm:mso-main}
There exists a constant-size $\MSO_2$ formula $\psi$, such that checking whether $\psi$ is true in a semi-complete digraph of constant cutwidth and pathwidth is \NP-hard.
\end{theorem}

\medskip
\noindent\textbf{Outline.} The remaining part of the paper is organized as follows. In Section~\ref{sec:prelim} we introduce all the necessary notions and recall various tools that will be used. Section~\ref{sec:main} contains descriptions of the proofs of Theorems~\ref{thm:bundle},~\ref{thm:main},~\ref{thm:immersion} and~\ref{thm:immersionclosed}. Finally, in Section~\ref{sec:conclusion} we gather conclusions and deliberate on new challenges arising from the current work. The proof of Theorem~\ref{thm:mso-main} has been deferred to Appendix~\ref{sec:mso}, which contains a detailed description of tools from logic applicable in our setting.

\section{Preliminaries}\label{sec:prelim}

Let $G=(V,E)$ be a finite digraph. Digraphs can have loops and multiple edges.  We refer to the book of Bang-Jensen  and Gutin  \cite{BangG089_book}  for an introduction to directed graphs. We use $V(G)$ and $E(G)$ to represent the vertex
set and arc set of $G$, respectively. By the size of $G$ we denote $|G|=|V|+|E|$.  Given a subset $V'\subseteq
V(G)$ of   $G$, let $G[V']$ denote the digraph induced by
$V'$.   A vertex $u$ of $G$ is an {\em in-neighbour} ({\em
out-neighbour}) of a vertex $v$ if $(u,v)\in E$ ($(v,u)\in E$,
respectively). The {\em in-degree} $d^-(v)$ ({\em out-degree}
$d^+(v)$) of a vertex $v$ is the number of its in-neighbours
(out-neighbours). We denote by $N^-(v)$ and by $N^+(v)$ the set of all in-neighbours
and out-neighbours of $v$ correspondingly.

A digraph is \emph{simple} if it has no loops, and for every pair of distinct vertices $u,v$ there is at most one arc $(u,v)$.
A digraph is \emph{semi-complete} if it is simple and for every pair of distinct vertices $u,v$ there is at least one arc with endpoints $u,v$.
A digraph is a \emph{tournament} if it is simple and for every pair of distinct vertices $u,v$ there is exactly one arc with endpoints $u,v$.

\medskip\noindent\textbf{Topological containment and immersions.}
Let $H,G$ be digraphs. We say that mapping $\eta : V(H)\cup E(H) \to V(G)\cup E(G)$ is a {\em{model}} of $H$ in $G$, if the following conditions are satisfied:
\begin{itemize}
\item $\eta(V(H))\subseteq V(G)$;
\item for every edge $(u,v)\in E(H)$, $\eta((u,v))$ is a directed path leading from $\eta(u)$ to $\eta(v)$;
\item for every loop $(u,u) \in E(H)$, $\eta((u,u))$ is a directed cycle passing through $\eta(u)$.
\end{itemize}
By imposing further conditions on the model we obtain various notions of containment for digraphs. If we require that $\eta$ is surjective on $V(H)$ and the paths in $\eta(E(H))$ are internally vertex-disjoint, we obtain the notion of {\emph{topological containment}}; in this case we say that $\eta$ is an {\emph{expansion}} of $H$ in $G$. If we relax the condition of paths being internally vertex-disjoint to being edge-disjoint, we arrive at the notion of {\emph{immersion}}; then $\eta$ is an {\emph{immersion}}\footnote{Sometimes in the literature this notion is called a {\emph{weak immersion}}. We remark that all our results can be easily extended to strong immersions as well.} of $H$ into $G$.  

In our setting we are given digraph $H$ (with possible loops and multiple edges) and a semi-complete digraph $T$. We ask whether $H$ is topologically contained (immersible) in $T$. Observe that without changing the answer for these questions we can subdivide once every loop in $H$. Thus we may assume that $H$ does not contain any loops; though, it may have multiple edges.

We say that $(G;v_1, \dots, v_h)$ is a \emph{rooted} digraph if $G$ is digraph and $v_1, \dots, v_h$ are distinct vertices of $V(G)$. The notion of immersion can be generalized to rooted graphs.
Immersion $\eta$ is an immersion from rooted digraph  $\mathbf{H}=(H;u_1, \dots, u_h)$ to rooted digraph  $\mathbf{G}=(G;v_1, \dots, v_h)$ if
$\eta(u_i)=v_i$ for $i\in \{1,\dots, h\}$. Such an immersion is called an \emph{$\mathbf{H}$-immersion} or a \emph{rooted immersion}.

We study the following parameterized topological containment and immersion problems.

\defparproblem{\textsc{Topological containment}}{Digraphs $H$ and $G$}{$|H|$}{Is there expansion of $H$ in  $G$?}

\defparproblem{$\mathbf{H}$-\textsc{immersion (Rooted Immersion)}}{Rooted digraphs $\mathbf{H}$ and $\mathbf{G}$}{$|\mathbf{H}|$}{Is there $\mathbf{H}$-immersion of $\mathbf{H}$ in $\mathbf{G}$?}

\medskip\noindent\textbf{Separations and pathwidth.}

\begin{definition}\label{def:partition}
Let $T=(V,E)$ be a  digraph. For $A, B \subseteq V$, the pair $(A,B)$ is a separation of order $k$, if
\begin{itemize}
\item $A\cup B=V$,
\item $|A\cap B|\leq k$, and
\item there are no edges from $A\setminus B$ to $B\setminus A$.
\end{itemize}
The set  $A \cap B$ is the \emph{cut} corresponding to the separation $(A, B)$. Two separations $(A, B$) and $(C,D)$ \emph{cross} unless one of the following holds:
\begin{itemize}
\item $C\subseteq A$ and $B\subseteq  D$
 \item $ A \subseteq  C$ and $D \subseteq  B$.
 \end{itemize}
A set of separations is \emph{cross-free} if no two of its members cross.
\end{definition}
The following lemma is due to Fradkin and Seymour  \cite[(1.2)]{fradkin-seymour}.
\begin{lemma}\label{lem:order}
Let $G$ be a digraph and $t \geq 0$ be a integer. Let $S$ be a cross-free set of separations with $|S| = t$. Then the members of $S$ can be ordered $(A_1,B_1),\dots,(A_t,B_t)$ such that $A_1 \subseteq \cdots \subseteq  A_t $ and $B_t \subseteq \cdots \subseteq B_1$.
\end{lemma}

\begin{definition}\label{def:pw}
Given a digraph $G=(V,E)$, a sequence $W = (W_1, \dots , W_r)$ of subsets of $V$  is a \emph{path decomposition of $G$} if the following conditions are satisfied:
 \begin{itemize}
\item[(i)] $\bigcup_{1\leq i\leq r} W_i = V $;
 \item[(ii)] $W_i \cap W_k \subseteq W_j$ for $1 \leq i < j < k \leq r$;
\item[(iii)]  $\forall$  $(u,v)\in E$, either $u,v\in W_i$ for some $i$ or $u\in W_i$, $v\in W_j$ for some $i>j$.  \end{itemize}
 We call $W_1,\dots , W_r$ the \emph{bags} of the path decomposition.
 The \emph{width} of a path decomposition is equal to
$\max_{1\leq i \leq r} (|W_i| - 1)$; the \emph{pathwidth} of $G$ is the minimum width over all path decompositions of $G$.
\end{definition}

\noindent\textbf{Jungles and triples.} For a semi-complete digraph $T$, a set $Z\subseteq V(T)$ is called a {\emph{$k$-jungle}} if $|Z|=k$ and for every $x,y\in Z$ there does not exist a separation $(A,B)$ of order smaller than $k$ that {\emph{separates}} $y$ from $x$, i.e., such that $x\in A\setminus B$ and $y\in B\setminus A$. A triple of disjoint subsets $(A,B,C)$ is called a $k$-triple if $|A|=|B|=|C|=k$ and there exist orderings $(a_1,\ldots,a_k)$, $(b_1,\ldots,b_k)$, $(c_1,\ldots,c_k)$ of $A,B,C$, respectively, such that for all indices $1\leq i,j\leq k$ we have $(a_i,b_j),(b_i,c_j)\in E(G)$ and for each index $1\leq i\leq k$ we have $(c_i,a_i)\in E(G)$. The following result is an algorithmic version of the result of Fradkin and Seymour~\cite[(2.6)]{fradkin-seymour}. The proof closely follows the lines of their argumentation; however, for the sake of completeness we include it in Appendix~\ref{sec:proofs}.
\begin{lemma}\label{lem:jungle_triple_algo}
There exists an elementary function $f$, such that for every $k\geq 1$ and every semi-complete graph $T$ given together with its $f(k)$-jungle, it is possible to construct a $k$-triple in $T$ in time $O(|V(T)|^3\log |V(T)|)$.
\end{lemma}
It is easy to check that both a $k$-triple and a $k$-jungle are obstacles to admitting a path decomposition of width less than $k-1$; there must be a bag containing the whole part $B$ or the whole jungle, respectively. The following simple observation is also due to Fradkin and Seymour~\cite{fradkin-seymour}.
\begin{lemma}\label{lem:triple-model}
Every digraph $H$ is topologically contained in a $|H|$-triple.
\end{lemma}

 \noindent\textbf{Tools from logic.} We gather a detailed discussion on the applicability of tools from logic in our setting in Appendix~\ref{sec:mso}. Essentially, we observe that as on semi-complete graphs cliquewidth can be bounded by a function of pathwidth, and model checking of formulas of the logic $\MSO_1$ can be done by a fixed-parameter algorithm, parameterized by the length of the formula and cliquewidth, we get the following powerful meta-tool for constructing algorithms on path decomposition. For appropriate definitions we invite an interested reader to Appendices~\ref{sec:parameters} and~\ref{sec:mso}.

\begin{theorem}\label{thm:MC}
There exist an algorithm with running time $f(||\ff||,p)|V(T)|^2$ that given an $\MSO_1$ formula $\ff$ checks whether $\ff$ is satisfied in a semi-complete digraph $T$, given together with a path decomposition of width $p$.
\end{theorem}

\section{Main results}\label{sec:main}

\subsection{Approximation of pathwidth: proof of Theorem~\ref{thm:bundle}}

In this section we give an approximation algorithm computing the pathwidth of a semi-complete digraph. As the algorithm makes use of the same concepts as the technique of Fradkin and Seymour~\cite{fradkin-seymour}, we need to introduce the essential terminology.

Let $(A,B)$, $(C,D)$ be separations of $G$ that do not cross and suppose they have orders $ i,j$, respectively. Without loss of generality, suppose that $A \subseteq C$ and $D \subseteq B$. We say that $(A, B)$ and $(C,D)$ are \emph{$k$-close} if $(B\setminus A)\cap(C\setminus D)<k |i-j|$. The notion of a {\emph{bundle}} is crucial for computing the pathwidth of a semi-complete digraph.
\begin{definition}[\textbf{Bundle}]\label{def:bundle}
For an integer $k > 0$, a \emph{bundle} in $G$ of order $k$ is a cross-free set $\Bb$ of separations of $G$, each of order $<k$, such that
\begin{itemize}
\item[(i)] no two members of $\Bb$ are $k$-close;
\item[(ii)] if $(A, B)$ is a separation of $G$ of order $i < k$, then one of the following holds:
\subitem(a) $(A,B)\in \Bb$
\subitem(b) $(A,B)$ crosses some $(C,D) \in \Bb$ of order $\leq  i$
\subitem(c) $(A,B)$ is $k$-close to some $(C,D) \in \Bb$ of order $\leq  i$.
\end{itemize}
 \end{definition}

A separation $(A,B)$ of $G$ is an \emph{$\ell$-separator of $U\subseteq V(H)$} if $|U\setminus A|$, $|U\setminus B| \geq \ell$.
We say that $U$ is \emph{$(k, \ell)$-separable} if there exists an $\ell$-separator of $U$ of order at most $k$.

The following proposition implicitly follows from the proof of (2.11) in~\cite{fradkin-seymour}.
\begin{lemma}\label{lem:maximality_bundle}
Let  $\Bb$ be a bundle of order $k$ and let $(A_1,B_1),(A_2,B_2),\ldots,(A_m,B_m)$ be the members of  $\Bb$ ordered such that
$A_1\subseteq \cdots \subseteq A_m$ and $B_m\subseteq \cdots \subseteq B_1$. If for some
$1\leq j \leq m$,  $|A_{j+1}\cap B_j|>4k^2+7k$, then the set
$(A_{j+1}\setminus B_{j+1})\cap (B_j\setminus A_j)$ is not $(k,k^2)$-separable.
\end{lemma}

We are now ready to prove Theorem~\ref{thm:bundle}.

\begin{proof}[Proof of Theorem~\ref{thm:bundle}] The first steps of ours are similar to the approach of Fradkin and Seymour~\cite{fradkin-seymour}; however, we need to use new techniques to obtain a substantially better running time.
  First we construct a bundle $\Bb$ of order $k$. Let $(A_1,B_1),(A_2,B_2),\ldots,(A_m,B_m)$, $A_1=B_m=\emptyset$ and  $A_m=B_1=V(T)$, be the ordering of the separations in $\Bb$ given by Lemma~\ref{lem:order}. We construct a path decomposition $W= (W_1, \dots , W_{m-1})$ consisting of bags $W_i=A_{i+1}\cap B_i$ for $i\in \{1,2,\ldots,m-1\}$. 
   If the width of $W$ is at most $4k^2+7k$, then the algorithm outputs $W$ and terminates. Otherwise, there is such an index $j$, for which $|A_{j+1}\cap B_j|>4k^2+7k$. As $|A_j\cap B_j|, |A_{j+1}\cap B_{j+1}|<k$, then $|(A_{j+1}\setminus B_{j+1})\cap (B_j\setminus A_j)|>4k^2+5k$. By Lemma~\ref{lem:maximality_bundle},
 $(A_{j+1}\setminus B_{j+1})\cap (B_j\setminus A_j)$ is not $(k,k^2)$-separable. Then  we find a $k$-jungle $\Jj$ inside $(A_{j+1}\setminus B_{j+1})\cap (B_j\setminus A_j)$.

Therefore, to prove the theorem, it is sufficient to show how to
\begin{itemize}
\item[(a)] Construct a bundle of order $k$ in time $2^{O(k\log k)}|V|^3\log |V|$, and
\item[(b)] Construct in time $O(|W|^2)$ a $k$-jungle from a set $W$ of cardinality at least $5k+4\ell$, providing it is not $(k,\ell)$-separable.
\end{itemize}

\medskip

\noindent{\bf{(a) Construction of a bundle.}} We will inductively construct a sequence of more and more refined cross-free families of separations $\Bb_1,\Bb_2,\ldots,\Bb_k$, such that following conditions are satisfied:
\begin{itemize}
\item $\Bb_i$ is a cross-free family of separations of order $<i$;
\item $\Bb_i\subseteq \Bb_j$ for $i\leq j$;
\item no two separations in $\Bb_i$ are $k$-close;
\item for every $i\leq k$, every separation of $T$ of order $<i$ either belongs to $\Bb_i$, crosses some separation from $\Bb_i$ or is $k$-close to some separation from $\Bb_i$.
\end{itemize}
These conditions imply that $\Bb_k$ is a bundle of order $k$. Note that families $\Bb_i$ are not necessarily bundles of order $i$ for $i<k$, as the closeness conditions refer to the value of $k$ instead of $i$.

We begin with constructing the family $\Bb_1$ consisting of separations of order $0$. In $O(|V|^2)$ time we partition $T$ into strongly connected components; as $T$ is a tournament, the acyclic digraph of strongly connected components is a linear order. Let $F_1,F_2,\ldots,F_h$ be these components, ordered with respect to this linear order. As $\Bb_1$ we take the family of separations of order $0$ defined as follows: $\Bb_1:=\{(\bigcup_{i=1}^j F_j,\bigcup_{i=j+1}^h F_j): j\in\{0,1,2,\ldots,h\}\}$. It is easy to check that conditions imposed on $\Bb_1$ are satisfied.

Now we present an algorithm that given family $\Bb_i$ constructs family $\Bb_{i+1}$. We take $\Bb_{i+1}:=\Bb_i$ and greedily incorporate new separations of order $i$ into $\Bb_{i+1}$, such that every new separation does not cross older ones and is not $k$-close to any of them. Having obtained such a maximal family $\Bb_{i+1}$, i.e., one, into which no new separation can be incorporated, we conclude that all the conditions imposed on $\Bb_{i+1}$ are satisfied.

Let $(A_1,B_1),\ldots,(A_h,B_h)$ be the separations at some step of the construction of $\Bb_{i+1}$ ordered by an ordering obtained from Lemma~\ref{lem:order}. Note that $(A_1,B_1)=(\emptyset,V(T))$ and $(A_h,B_h)=(V(T),\emptyset)$; therefore, every new separation has to be inserted between two consecutive separations that are already in $\Bb_{i+1}$. We try to insert the new separation between $(A_j,B_j)$ and $(A_{j+1},B_{j+1})$ for every $j\in \{1,2,\ldots,h-1\}$. If we do not succeed for any $j$, we conclude that $\Bb_{i+1}$ is already maximal.

Consider inserting a new separation $(C,D)$ between $(A_j,B_j)$ and $(A_{j+1},B_{j+1})$. We have the following simple claim:
\begin{claim}\label{cl:trivial}
If $(C,D)$ is neither $k$-close to $(A_j,B_j)$ nor to $(A_{j+1},B_{j+1})$, then it is not $k$-close to any separation $(A_r,B_r)\in \Bb_{i+1}$.
\end{claim}
The proof of Claim~\ref{cl:trivial} uses standard triangle inequality properties of set operations, and can be found in Appendix~\ref{sec:proofs}. Hence, we can only check that the new separation is not $k$-close to $(A_j,B_j)$ nor to $(A_{j+1},B_{j+1})$. Consider $S=T[B_j\cap A_{j+1}]$. In this semi-complete digraph we have two sets of terminals: $X=A_j\cap B_j$ and $Y=A_{j+1}\cap B_{j+1}$. We need to determine, whether there exists a separation $(C',D')$ of $S$ of order $i$ that separates $X$ from $Y$, is not $k$-close to the separation $(X,V(S))$ nor to the separation $(V(S),Y)$, and is distinct from both of them. Every such a separation $(C',D')$ corresponds to an objective separation $(C,D)=(C'\cup A_j,D'\cup B_{j+1})$ and vice versa: every objective separation $(C,D)$ induces such a separation $(C',D')=(C\cap V(S),D\cap V(S))$ in $S$. Therefore, we need to solve the following problem in the subdigraph $S=T[B_j\cap A_{j+1}]$, for $X=A_j\cap B_j$, $Y=A_{j+1}\cap B_{j+1}$, $a=\min(k(i-|X|),1)$, $b=i$ and $c=\min(k(i-|Y|),1)$.

\defparproblem{\toursection}{A semi-complete digraph $S$, subsets of vertices $X,Y\subseteq V(S)$ and integers $a,b,c$}{$a+b+c$}{Find a separation $(C,D)$ of $S$ such that {\emph{(i)}} $X\subseteq C$, {\emph{(ii)}} $Y\subseteq D$, {\emph{(iii)}} $|C\setminus D|\geq a$, {\emph{(iv)}} $|D\setminus C|\geq c$, {\emph{(v)}} $|C\cap D|\leq b$, or correctly claim that no such exists.}

We show an algorithm that solves \toursection in time \linebreak $2^{O(\min(a+c,b)\log (a+b+c))}|V(S)|^2\log |V(S)|$. Before we proceed to its presentation, let us shortly discuss why it will finish the proof of Theorem~\ref{thm:bundle}. We only need to show that the number of applications of this subroutine is bounded by $O(k|V(T)|)$, as for each call we have $|V(S)|\leq |V(T)|$ and $O(\min(a+c,b)\log (a+b+c))=O(k\log k)$. Note that the number of calls with a positive result, i.e., those that finished with finding a feasible separation, is bounded by $O(|V(T)|)$ for the whole run of the algorithm, as each such result implies incorporating one new separation to the constructed path decomposition. Moreover, for every considered $i$ the number of calls with negative results obtained while considering $i$ is bounded by the size of the bundle after this step, thus also by $O(|V(T)|)$. It follows that the total number of calls is bounded by $O(k|V(T)|)$ and we are done.

The algorithm for \toursection is based on the technique of {\emph{color coding}}, a classical tool in parameterized complexity introduced by Alon et al.~\cite{AlonYZ95}. Although the intuition behind color coding is probabilistic, it is well-known that almost all algorithms designed using this approach can be derandomized using the technique of {\emph{splitters}} of Naor et al.~\cite{naor-schulman-srinivasan-derandom}. We find it more convenient to present our algorithm already in the derandomized version. We will use the following folklore abstraction of derandomization of random coloring; for the sake of completeness, the proof can be found in Appendix~\ref{sec:proofs}.

\begin{lemma}\label{lem:derandomize}
There exists an algorithm, which given integers $p,q$ and universe $U$, in time \linebreak $2^{O(\min(p,q)\log (p+q))}|U| \log |U|$ returns a family $\randfamily$ of subsets of $U$ of size $2^{O(\min(p,q)\log (p+q))} \log |U|$ with the following property: for every subsets $A,B\subseteq U$, such that $A\cap B=\emptyset$, $|A|\leq p$ and $|B|\leq q$, there exists a set $R\in \randfamily$ such that $A\subseteq R$ and $B\cap R=\emptyset$.
\end{lemma}
 
First, using Lemma~\ref{lem:derandomize} in time $2^{O(\min(a+c,b)\log (a+b+c))} |V(S)|\log |V(S)|$ we generate a family $\randfamily$ for universe $U=V(S)$ and constants $p=a+c$, $q=b$; note that $|\randfamily|\leq 2^{O(\min(a+c,b)\log (a+b+c))} \log |V(S)|$. Properties of $\randfamily$ ensure that for every separation $(C,D)$ that satisfies properties given in the problem statement, there exists a set $R\in \randfamily$ such that $|(C\setminus D)\cap R|\geq a$, $|(D\setminus C)\cap R|\geq c$ and $(C\cap D)\cap R=\emptyset$. In such a situation we say that $R$ is {\emph{compatible}} with separation $(C,D)$. We iterate through the family $\randfamily$, for each $R\in\randfamily$ trying to find a separation that is compatible with it; by Lemma~\ref{lem:derandomize}, for at least one choice we are able to find a feasible separation.

Consider the subdigraph $S_R=S[R]$. In $O(|V(S)|^2)$ time partition $S_R$ into strongly connected components. As $S_R$ is semi-complete, they can be naturally arranged in a linear order, such that the arcs between different components are directed only to the one lower in the order. Let us denote this order of the components by $(F_1,F_2,\ldots,F_g)$. Assume there exists a feasible separation $(C,D)$ that is compatible with $R$. As $R$ is disjoint with $C\cap D$, we infer that there exists an index $h$ such that $\bigcup_{i=1}^{h} V(F_i)\subseteq C\setminus D$ and $\bigcup_{i=h+1}^{g} V(F_i)\subseteq D\setminus C$, i.e., a prefix of the order is entirely contained in one side of the separation, and the corresponding suffix is entirely contained in the second side. Let $h_1$ be the smallest index such that $|\bigcup_{i=1}^{h_1} V(F_i)|\geq a$ and $X\subseteq \bigcup_{i=1}^{h_1} V(F_i)$. Clearly, $h_1\leq h$ as otherwise either there would be less than $a$ vertices in the intersection of $R$ and $C\setminus D$, or there would be a vertex of $X$ in $D\setminus C$. Similarly, if by $h_2$ we denote the largest index such that $|\bigcup_{i=h_2}^{g} V(F_i)|\geq c$ and $Y\subseteq \bigcup_{i=h_2}^{g} V(F_i)$, then $h_2\geq h$. Note that $h_1,h_2$ can be computed in $O(|V(S)|)$ time and if $h_2<h_1$, then we may terminate the computation for this particular choice of $R$. On the other hand, if $h_2\geq h_1$ then every separation $(C,D)$ with $|C\cap D|\leq b$, such that $\bigcup_{i=1}^{h_1} V(F_i) \subseteq C\setminus D$ and $\bigcup_{i=h_2}^{g} V(F_i)\subseteq D\setminus C$, is sufficient for our needs. Therefore, all we need is to find the minimum vertex cut from $\bigcup_{i=1}^{h_1} V(F_i)$ to $\bigcup_{i=h_2}^{g} V(F_i)$, where vertices from $R$ have infinite capacity, while vertices from $V(S)\setminus R$ have unit capacities. This can be done in $O(b|V(S)|^2)$ time using the classical Ford-Fulkerson algorithm: note that once $b+1$ paths are found, we may terminate further computation, as no cut of small enough size can be found.

\medskip

\noindent{\bf{(b) Constructing a jungle.}} Assume we have a set $W$, such that $|W|\geq 5k+4\ell$ and $W$ is not $(k,\ell)$-separable. Our construction mimics the proof of (2.8) in~\cite{fradkin-seymour}.

Let $S=T[W]$. In time $O(|W|^2)$ we construct a set  $Z=\{w\in W\ |\ d_S^{+}(w)\geq k+\ell \wedge d_S^{-}(w)\geq k+\ell\}$. Note that for any semi-complete digraph and any integer $d$, the number of vertices of outdegree at most $d$ is bounded by $2d+1$, as otherwise already in the subdigraph induced by them there would be a vertex of higher outdegree. The same holds for indegrees. Hence, we infer that $|Z|\geq |W|-(2(k+\ell)-1)-(2(k+\ell)-1)\geq k$. Take any $x,y\in Z$. For the sake of contradiction, assume that there is a separation $(A,B)$ of order less than $k$ that separates $x$ from $y$. As $W$ is not $(k,\ell)$-separable, either $|(A\setminus B)\cap W|\leq \ell-1$ or $|(B\setminus A)\cap W|\leq \ell-1$. However, this means that either $x$ or $y$ has in-degree or out-degree in $S$ bounded by $k+\ell-1$, a contradiction with $x,y\in Z$. Therefore, every subset of $Z$ of cardinality $k$ is the desired $k$-jungle $\Jj$.
\end{proof}

\subsection{Irrelevant vertex procedure for rooted immersion}

In this section we show the crucial ingredient needed for the application of the framework to Theorem~\ref{thm:immersion}. As we consider the rooted version of the problem, in the  case of finding a large obstacle we cannot immediately provide an answer. Instead, we make use of the so-called {\emph{irrelevant vertex technique}}: we identify a vertex in the obstacle that can be safely removed from the digraph. In Section~\ref{subsec:sketches} we show how the following result can be used in the proof of Theorem~\ref{thm:immersion}.

\begin{lemma}\label{thm:irrelevant}
Let $\Ii=((H;u_1,u_2,\ldots,u_t),(T;v_1,v_2,\ldots,v_t))$ be an instance of the \himmersion{} problem, where $T$ is a semi-complete graph containing a $p(k)$-triple $(A,B,C)$ disjoint with $\{v_1,v_2,\ldots,v_t\}$ and $k=|H|$. Then it is possible in time $O(p(k)^2|V(T)|^2)$ to identify a vertex $x\in B$ such that $\Ii$ is a YES instance of \himmersion{} if and only if $\Ii'=((H;u_1,u_2,\ldots,u_t),(T\setminus\{x\};v_1,v_2,\ldots,v_t))$ is a YES instance.
\end{lemma}

In our proof $p(x)= 80x^2+80x+5$. The proof is quite long and technical, so here we present only a brief, intuitive sketch. The full, formal version can be found in Appendix~\ref{sec:irrelevant}.

Recall that in the triple $(A,B,C)$ we have arcs from every vertex of $A$ to every vertex of $B$, from every vertex of $B$ to every vertex of $C$, and there exists a matching between $C$ and $A$ that consists of arcs directed from $C$ to $A$. Let us assume that $\Ii$ is a YES instance and let $\eta$ be the corresponding rooted immersion of $H$ such that $\Qq=\eta(E(H))$, the family of paths between corresponding images of vertices of $H$, has minimum total sum of lengths. To ease the presentation, we assume that none of vertices of the triple is in $\eta(V(H))$. The goal is to find a vertex $x$ in the part $B$, such that all the paths going through $x$ can be rerouted around --- thus we prove that if there exists some solution, then there exists also a solution that does not use $x$. We first observe some structural properties of the family $\Qq$ that follow from its minimality. The intuition is that if a path from $\Qq$ visits one of the parts $A,B,C$ multiple times, then one can use some unused parts of the triple to create a shortcut between the first and the last visit. All in all, one can show that at most $4k$ vertices from $B$ are visited by some path from $\Qq$, and at most $O(k^2)$ arcs of the matching between $C$ and $A$ have an endpoint visited by some path of $\Qq$. Thus, by choosing the size of the triple to be sufficiently large, one may assume that we are given an arbitrary large reservoir of {\emph{free}} vertices from $B$ and {\emph{free}} arcs of the matching between $C$ and $A$ that can be used for rerouting.

Consider vertex $x\in B$. The goal is to modify the family $\Qq$ to $\Qq''$ so that {\emph{(i)}} if a path enters $x$, then it enters $x$ from $A$; {\emph{(ii)}} if a path leaves $x$, then it leaves $x$ to $C$. Once both properties are satisfied, we can rerout all the paths going through $x$ to some other free vertex from $B$. Of course, this can be not possible for all vertices $x$. We find a vertex for which this is possible in two phases. First, we find a sufficiently large set $X\subseteq B$, such that for each $x\in X$ one can find an intermediate family $\Qq'$ for which just {\emph{(i)}} is satisfied. Second, inside $X$ we find a single vertex $x$, for which $\Qq'$ can be further modified to $\Qq''$ satisfying both {\emph{(i)}} and {\emph{(ii)}}. Note that we have to keep track of the number of free vertices and free edges used during the first phase so that the second phase has still enough space for reroutings. We now explain only how the first phase can be performed --- the second phase follows the same scheme.

For each $b\in B$, we partition the inneighbors of $b$ outside $A$ into two sets. One set, denoted $R_b$, contains these inneighbors, which have a sufficiently large number of outneighbors in $B$. We observe that every path entering $b$ from $R_b$ can be easily rerouted, because each vertex of $R_b$ must have a large number of free outneighbors in $B$, so we can redirect the path via a free outneighbor in $B$ and around the whole triple back to $b$, using one free arc of the matching between $C$ and $A$. The second set, denoted $G_b$, contains vertices that have small number of outneighbors in $B$ and are hence problematic. However, if sufficiently many vertices $b$ have empty $G_b$, we can already output these vertices as a feasible $X$.


Otherwise, we consider the vertices with nonempty $G_b$ and we know that we have a large number of them. We define an auxiliary digraph $S$ on the set of vertices $b\in B$ with nonempty $G_b$. We put $(b,b')\in E(S)$ iff for every $v\in G_b$ either $v\in G_{b'}$ or there exists $v'\in G_{b'}$ such that $(v,v')\in E(T)$. Here comes the main trick of the proof: the fact that $T$ is semi-complete implies that so do $S$. Now consider a vertex $b\in B$ with high outdegree in $S$. We show how a path entering $b$ from $v\in G_b$ can be rerouted so that it enters $b$ from $A$. As $b$ has large outdegree in $S$, for many $b'\in B$ which are still free we can access some $v'\in G_{b'}$ either by staying in $v$ (if $v\in G_{b'}$) or using one arc. If $v=v'$ for some $b'$, we can follow the same strategy as for the sets $R_b$. Otherwise, we need to check that some of the arcs accessing vertices $v'$ are not yet used by other paths; this can be problematic as vertices $v'$ are not necessarily distinct. However, we know that each of them belongs to some $G_{b'}$, so they can appear only in limited number of copies as they have only limited number of outneighbors in $B$. Hence, we can distinguish sufficiently many pairwise distinct ones to be sure that one of the arcs accessing them is not yet used. Therefore, we can replace usage of the arc $(v,b)$ with first accessing some other $G_{b'}$, then the free vertex $b'$, and then moving around the triple using a free arc of the matching between $C$ and $A$ (see Figure~\ref{fig:irrelevant} in Appendix~\ref{sec:irrelevant} for an illustration). As $S$ is semi-complete and sufficiently large, there exists a sufficiently large number of vertices $b$ with sufficiently large outdegree in $S$, from which the set $X$ can be formed.

\subsection{Sketches of proofs of Theorems~\ref{thm:main},~\ref{thm:immersion} and~\ref{thm:immersionclosed}}\label{subsec:sketches}

In this Section we sketch how Theorems~\ref{thm:main},~\ref{thm:immersion} and~\ref{thm:immersionclosed} follow from an application of the WIN/WIN approach. The proofs are essentially arranging already known pieces together, so we just discuss the main ideas. Fully formal descriptions can be found in Appendix~\ref{sec:main-proofs}.

Theorem~\ref{thm:main} follows almost directly from Theorems~\ref{thm:bundle},~\ref{thm:MC} and Lemmas~\ref{lem:jungle_triple_algo}, \ref{lem:triple-model}. We first run approximation of pathwidth. If pathwidth is too large, we are guaranteed that the digraph contains a sufficiently large triple, which topologically contains every digraph of given size. Hence, we may provide a positive answer. Otherwise, we have a path decomposition of small width, on which we can run the algorithm given by Theorem~\ref{thm:MC}. One just needs to observe that containing $H$ topologically can be expressed in $\MSO_1$: the formula existentially quantifies over images of vertices of $H$, and then checks existence of disjoint paths as disjoint subsets of vertices in which endpoints cannot by separated by separation of order $0$. We remark that the quantifier alternation of this formula is constant, so the obtained running time dependence on the size of $H$ is elementary.

For Theorems~\ref{thm:immersion} and~\ref{thm:immersionclosed} it will be convenient to use the following surprising result of Ganian et al.~\cite{GanianHO11}: for every $\ell$ there exists an $\MSO_1$ formula $\pi_\ell(s_1,s_2,\ldots,s_\ell, t_1,t_2,\ldots,t_\ell)$ that for a digraph with distinguished vertices $s_1,s_2,\ldots,s_\ell, t_1,t_2,\ldots,t_\ell$ (some of which are possibly equal) asserts whether there exists a family of edge-disjoint paths $P_1,P_2,\ldots,P_\ell$ such that $P_i$ begins in $s_i$ and ends in $t_i$ for $i=1,2,\ldots,\ell$. For Theorem~\ref{thm:immersion} we proceed similarly as before. After approximating pathwidth, we either obtain a large triple, in which using Theorem~\ref{thm:irrelevant} we can find an irrelevant vertex that can be safely removed, or obtain a path decomposition of small width, on which we can run algorithm given by Theorem~\ref{thm:MC} for an appropriate adaptation of formula $\pi_{|E(H)|}$. We obtain an additional $O(|V(T)|)$ overhead in the running time, as we may $O(|V(T)|)$ time run the approximation algorithm and remove one vertex. Unfortunately, the quantifier alternation of formula $\pi_\ell$ depends on $\ell$, so we do not obtain elementary dependence of the running time on the size of immersed graph. However, one can explicitly construct a dynamic program on path decomposition with elementary dependence. We provide such an algorithm in Appendix~\ref{sec:app-DP}. Using the same technique one can also obtain a dynamic program for topological containment with similar running time bounds; we omit the details. We remark that application of the formulas $\pi_\ell$ in the context of immersion in semi-complete digraphs has been already observed by Ganian et al.~\cite{GanianHO11}. 

Theorem~\ref{thm:immersionclosed} follows from combining this approach with the result of Chudnovsky and Seymour that immersion order on semi-complete tournaments is a well-quasi-order~\cite{ChudnovskyS11}. By~\cite{ChudnovskyS11}, class $\Pi$ can be characterized by a finite set of forbidden immersions $\{H_1,H_2,\ldots,H_r\}$. By a simple adaptation of formulas $\pi_\ell$, for each $H_i$ we can construct an $\MSO_1$ formula with one free monadic variable $X$ that asserts that after removing $X$ the graph does not admit $H_i$ as an immersion. By taking conjunction of these formulas and quantifying $X$ existentially we obtain an $\MSO_1$ formula that asserts belonging to $\Pi+kv$. In order to finish the proof it suffices to observe that the digraphs from the class $\Pi+kv$ have bounded pathwidth: a digraph with large pathwidth contains a big triple, which after removing at most $k$ vertices still contains every sufficiently small graph as an immersion. We run the algorithm from Theorem~\ref{thm:bundle}. If we obtain a large obstacle, we provide a negative answer; otherwise, we run the algorithm given by Theorem~\ref{thm:MC}.

\section{Conclusions and open problems}\label{sec:conclusion}
In this paper we showed that \textsc{Topological containment} and $\mathbf{H}$-\textsc{immersion} are FPT on semi-complete graphs.  A natural common generalization of 
\textsc{Topological containment} and $k$-\textsc{Vertex Disjoint Paths} is the rooted version of \textsc{Topological containment}, where specified vertices of $H$ should be mapped into specified vertices of $G$. 
The obvious next step of research would be to try to extend the results to rooted topological containment. Chudnovsky, Scott and Seymour in \cite{ChudnovskySS2011} proved that $k$-\textsc{Vertex Disjoint Paths} is in class \XP on semi-complete, i.e., solvable in polynomial time for fixed $k$.  This also implies that the rooted version of topological containment is in \XP on semi-complete graphs. 

Unfortunately, our approach used for rooted immersion does not apply to this problem. Dynamic programming on path decomposition works fine but the problem is with  the irrelevant vertex arguments. Even for $k=2$ there exist tournaments that contain arbitrarily large triples, but in which every vertex is relevant; we provide such an example in Appendix~\ref{sec:app-counter}. This suggests that a possible way to obtaining an \FPT algorithm for \vertexpath problem requires another width parameter admitting more powerful obstacles.

\bibliographystyle{splncs03}
\bibliography{flat_tournament-containment}

\newpage
\appendix

\section{Omitted proofs from the main body}\label{sec:proofs}

\begin{proof}[Proof of Lemma~\ref{lem:jungle_triple_algo}]
For an integer $k$, let $R(k,k)$ denote the Ramsey number, that is the smallest integer such that every red-blue coloring of the edges of the complete graph on $R(k,k)$ vertices contains a monochromatic clique of size $k$. By the theorem of Erd{\"o}s and Szekeres \cite{ErdosS35}, $R(k,k)\leq (1+o(1))\frac{4^{k-1}}{\sqrt{\pi k}}$. For $k\geq 1$, we define function the function $f$ as
$$f(k)=2^{12\cdot 2^{R(2k,2k)}}.$$
Moreover, let $r=R(2k,2k)$, $s=2^r$,  and $m=2^{12s}=f(k)$.

We say that a semi-complete digraph is {\emph{transitive}} if it contains a transitive tournament as a subdigraph. Every tournament with $m$ vertices contains a transitive tournament with $\log_{2}{m}$ vertices. Also an induced acyclic
subdigraph of an oriented (where there is at most one directed arc between a pair of vertices) $m$-vertex digraph
with $\log_{2}{m}$ vertices can be found in time $O(m^2\log m)$ \cite{RamanS06}. This algorithm can be modified into an algorithm finding a transitive semi-complete digraph in semi-complete digraphs by removing first all pairs of oppositely directed arcs, running the algorithm for oriented graphs, and then adding some of the deleted arcs to turn the acyclic digraph into a transitive semi-complete digraph.
Thus $T$ contains as a subdigraph a transitive semi-complete digraph, whose vertex set $X$ is a $12s$-jungle in $T$; moreover,
such a set $X$ can be found in time $O(|V(T)|^2\log |V(T)|)$.

The next step in the proof of Fradkin and Seymour is to partition the set $X$ into parts $X_1$ and $X_2$ of size $6s$ each such that
$X_1$ is complete to $X_2$, i.e., for each $x_1\in X_1$ and $x_2\in X_2$ we have $(x_1,x_2)\in E(T)$. Such a partition of the vertex set of the transitive semi-complete digraph can be easily found in time $O(|X|^2)$.
Because $X$ is a $12s$-jungle in $T$, there are at least $6s$ internally vertex-disjoint paths from $X_2$ to $X_1$ in $T$.
Let $R$ be a minimal induced subdigraph of $T$ such that $X \subseteq V(R)$ and there are $6s$ vertex-disjoint paths from $X_2$ to $X_1$ in $R$. Such a minimal subgraph $R$ can be found in time $O(|V(T)|^3\log |V(T)|)$ by repeatedly removing vertices $v\in V(T)\setminus X$ if there are $6s$ vertex-disjoint paths from $X_2$ to $X_1$ in $V(T)\setminus \{v\}$. As the subroutine for finding the paths we use the classical Ford-Fulkerson algorithm, where we finish the computation after finding $6s$ paths. Hence, the running time is $O(s|V(T)|^2)$. As we make at most $|V(T)|$ tests, and $s=O(\log |V(T)|)$, the claimed bound on the runtime follows.

Let $P_1, P_2, \dots, P_{6s}$ be vertex-disjoint paths from $X_2$ to $X_1$ in $R$. Fradkin and Seymour proved that the set of vertices $Q$ formed by the first two and the last two vertices of these $6s$ paths contains a $k$-triple. Thus, by checking every triple of subsets of $Q$ of size $k$ in time polynomial in $k$, we can find a $k$-triple. This step takes time $O\left(\binom{24s}{k}^3k^{O(1)}\right)=|V(T)|^{o(1)}$, as $s=O(\log |V(T)|)$ and $k=O(\log\log |V(T)|)$.
\end{proof}

\begin{proof}[Proof of Claim~\ref{cl:trivial} from the proof of Theorem~\ref{thm:bundle}]
We are to prove that if the new separation $(C,D)$ inserted between $(A_j,B_j)$ and $(A_{j+1},B_{j+1})$ is not $k$-close to $(A_j,B_j)$ nor to $(A_{j+1},B_{j+1})$, then it is not $k$-close to any separation $(A_r,B_r)\in \Bb_{i+1}$. By symmetry we assume that $r<j$.
Let us recall that $A_r\subseteq A_j\subseteq C\subseteq A_{j+1}$ and $B_r\supseteq B_j\supseteq D\supseteq B_{j+1}$.
As no separations already inside $\Bb_{i+1}$ are $k$-close, we have that
\begin{eqnarray*}
|(B_r\setminus A_r)\cap (A_j\setminus B_j)| & \geq & k\left| |A_r\cap B_r| - |A_j\cap B_j| \right|;\\
|(B_j\setminus A_j)\cap (C\setminus D)| & \geq & k\left| |A_j\cap B_j| - |C\cap D| \right|.
\end{eqnarray*}
Note that
\begin{eqnarray*}
\left((B_r\setminus A_r)\cap (A_j\setminus B_j)\right) \cup \left((B_j\setminus A_j)\cap (C\setminus D)\right) & \subseteq & \left((B_r\setminus A_r)\cap (C\setminus D)\right), \\
\left((B_r\setminus A_r)\cap (A_j\setminus B_j)\right) \cap \left((B_j\setminus A_j)\cap (C\setminus D)\right) & = & \emptyset.
\end{eqnarray*}
Therefore,
\begin{eqnarray*}
|(B_r\setminus A_r)\cap (C\setminus D)| & \geq & |(B_r\setminus A_r)\cap (A_j\setminus B_j)|+|(B_j\setminus A_j)\cap (C\setminus D)| \\
& \geq & k\left| |A_r\cap B_r| - |A_j\cap B_j| \right|+k\left| |A_j\cap B_j| - |C\cap D| \right| \\
& \geq & k\left| |A_r\cap B_r| - |C\cap D| \right|.
\end{eqnarray*}
This proves that $(A_r,B_r)$ and $(C,D)$ are not $k$-close.
\end{proof}

\begin{proof}[Proof of Lemma~\ref{lem:derandomize}]
Without loss of generality we may assume that $p\leq q$, as we can always replace every set in the constructed family with its complement, thus switching roles of $p$ and $q$. An $(n,r,r')${\emph{-splitter}} $\mathcal{S}$ is a family of functions from a universe $U$ of size $n$ to the set $\{1,2,\ldots,r'\}$ such that for every $X\subseteq U$, $|X|=r$, there exists $f\in \mathcal{S}$ that is injective on $X$. Naor et al.~\cite{naor-schulman-srinivasan-derandom} gave an explicit construction of $(n,r,r^2)$-splitter of size $O(r^6\log r\log n)$ in time $O(poly(r)\cdot n\log n)$. Using the algorithm of Naor et al. we construct a $(n,z,z^2)$-splitter $\mathcal{S}$ on $U$, where $z=\min(p+q,n)$. Then, for every $f\in \mathcal{S}$ and every $Z\subseteq \{1,2,\ldots,z^2\}$, $|Z|=p$, we add to the constructed family $\randfamily$ the set $f^{-1}(Z)$. Assume now that we have $A,B\subseteq U$ such that $|A|\leq p$ and $|B|\leq q$. Obtain $A'$ and $B'$ by adding arbitrary elements of $U\setminus (A\cup B)$ to $A$ and $B$ so that $|A'|+|B'|=z$. By definition of the splitter, there exists $f\in \mathcal{S}$ that is injective on $A'\cup B'$. To finish the proof one needs to observe that if we take $R=f^{-1}(f(A'))$, then $A\subseteq R$ and $B\cap R=\emptyset$. The bounds on the time complexity of the algorithm and on the size of $\randfamily$ follow from the bounds given by Naor et al. and from the fact that the number of subsets of $\{1,2,\ldots,z^2\}$ of size $p$ is bounded by $2^{O(p\log z)}\leq 2^{O(p\log (p+q))}$.
\end{proof}

\section{Full proofs of Theorems~\ref{thm:main},~\ref{thm:immersion} and~\ref{thm:immersionclosed}}\label{sec:main-proofs}

\begin{proof}[Proof of Theorem~\ref{thm:main}]
Let $f$ be the function given by Lemma~\ref{lem:jungle_triple_algo} and denote $k=f(|H|)$. We first apply the algorithm from Theorem~\ref{thm:bundle}, which in $2^{O(k\log{k})}|V|^3 \log{|V|}$ time finds either a $k$-jungle in $T$, or returns a path decomposition of width at most $4k^2+7k$. In the first case, by Lemma~\ref{lem:jungle_triple_algo}  we know that $T$ contains a $|H|$-triple. Lemma~\ref{lem:triple-model} asserts that we may safely provide a positive answer. Otherwise, we can observe that we can easily construct an $\MSO_1$ formula $\ff_H$ that checks, whether $H$ can is topologically contained in $T$. First, we existentially quantify over the images of vertices of $H$. Then, for every $(u,v)\in E(H)$ we quantify existence of a set $X_{(u,v)}$, such that all the sets $X_{(u,v)}$ are vertex disjoint and $X_{(u,v)}\cup \{\eta(u),\eta(v)\}$ does not admit a separation of order $0$ that separates $\eta(v)$ from $\eta(u)$. This asserts existence of required paths between vertices. We can conclude the proof by applying Theorem~\ref{thm:MC} to formula $\ff_H$ and semi-complete digraph $T$. In order to show that the obtained running time is elementary, one can to observe that the formula $\ff_H$ has constant quantifier alternation, so the dependence on $||\ff_H||$ and width of the decomposition of the running time of the algorithm given by Theorem~\ref{thm:MC} is elementary. Alternatively, one can construct the dynamic program on path decomposition by hand in a similar manner as the dynamic program for rooted immersion, given in Appendix~\ref{sec:app-DP}.
\end{proof}

\begin{proof}[Proof of Theorem~\ref{thm:immersion}]
The algorithm performs at most $|V(T)|$ iterations. In each iteration we either delete one vertex from the graph so that the answer does not change, or resolve the problem completely. The complexity of each step is $h(|H|)|V(T)|^3\log |V(T)|$, which gives a $h(|H|)|V(T)|^4\log |V(T)|$ bound on the total running time.

In each step we run algorithm from Theorem~\ref{thm:bundle} for $k=f(p(|H|))$ and semi-complete digraph $T\setminus \{v_1,v_2,\ldots,v_t\}$, which in time $2^{O(k\log k)}|V(T)|^3\log |V(T)|$ either outputs a path decomposition of $T\setminus \{v_1,v_2,\ldots,v_t\}$ of width at most $O(k^2)$, or a $k$-jungle in $T\setminus \{v_1,v_2,\ldots,v_t\}$. If we find a jungle using Theorem~\ref{lem:jungle_triple_algo} in $O(|V(T)|^3\log |V(T)|)$ time we can find a $p(k)$-triple in $T\setminus\{v_1,v_2,\ldots,v_t\}$ and, by further usage of Theorem~\ref{thm:irrelevant}, an irrelevant vertex in it. Theorem~\ref{thm:irrelevant} ensures that removal of this very vertex yields an equivalent instance.

However, if we obtain a decomposition, we can include all the excluded roots into every bag (note that there are at most $|V(H)|$ of them) to obtain a decomposition of $T$. Then we may proceed in two ways. Firstly, let us recall the result of Ganian et al.~\cite{GanianHO11}, that for every $\ell$ there exists an $\MSO_1$ formula $\pi_\ell(s_1,s_2,\ldots,s_\ell, t_1,t_2,\ldots,t_\ell)$ that for a digraph with distinguished vertices $s_1,s_2,\ldots,s_\ell, t_1,t_2,\ldots,t_\ell$ (some of which are possibly equal) asserts whether there exists a family of edge-disjoint paths $P_1,P_2,\ldots,P_\ell$ such that $P_i$ begins in $s_i$ and ends in $t_i$, for $i=1,2,\ldots,\ell$. We take formula \linebreak $\pi_{|E(H)|}(s_1,s_2,\ldots,s_{|E(H)|}, t_1,t_2,\ldots,t_{|E(H)|})$, additionally quantify existentially images of vertices of $H$ that are not roots, and replace variables $s_i,t_i$ with corresponding roots or quantified images. Thus we obtain a new formula $\ff_\mathbf{H}$ that asserts that $\mathbf{H}$ is a rooted immersion of $\mathbf{T}$. We may now use Theorem~\ref{thm:MC} to obtain the algorithm. Unfortunately, as the formula constructed by Ganian et al.~\cite{GanianHO11} has quantifier alternation depending on $|\mathbf{H}|$, the guarantees on the running time obtained in this manner are non-elementary. However, it is possible to construct a dynamic program on path decomposition by hand, which has elementary dependence on the width and $|\mathbf{H}|$. We provide the details in Appendix~\ref{sec:app-DP}.
\end{proof}

\begin{proof}[Proof of Theorem~\ref{thm:immersionclosed}]
As $\Pi$ is immersion-closed, by the result of Chudnovsky and Seymour~\cite{ChudnovskyS11} that immersion order on semi-complete tournaments is a well-quasi-order, we infer that $\Pi$ can be characterized by admitting none of a family of semi-complete digraphs $\{H_1,H_2,\ldots,H_r\}$ as an immersion, where $r=r(\Pi)$ depends only on the class $\Pi$. Let us recall the result of Ganian et al.~\cite{GanianHO11}, that for every $\ell$ there exists an $\MSO_1$ formula $\pi_\ell(s_1,s_2,\ldots,s_\ell, t_1,t_2,\ldots,t_\ell)$ that for a digraph with distinguished vertices $s_1,s_2,\ldots,s_\ell, t_1,t_2,\ldots,t_\ell$ (some of which are possibly equal) asserts whether there exists a family of edge-disjoint paths $P_1,P_2,\ldots,P_\ell$ such that $P_i$ begins in $s_i$ and ends in $t_i$ for $i=1,2,\ldots,\ell$. For every $i\in \{1,2,\ldots,r\}$ we construct an $\MSO_1$ formula $\ff_i(X)$ that is true if digraph $G\setminus X$ contains $H_i$ as an immersion: we simply quantify existentially over the images of vertices of $H_i$, use the appropriate formula $\pi_{|E(H)|}$ for quantified variables to express existence of paths, and at the end relativize the whole formula to the subdigraph induced by $V(T)\setminus X$. Hence, if we denote by $\psi_k(X)$ the assertion that $|X|\leq k$ (easily expressible in the first order logic by a formula, whose length depends on $k$), the formula $\ff=\exists_X \psi_k(X) \wedge \bigwedge_{i=1}^r \neg\ff_i(X)$ is true exactly in semi-complete digraphs, from which one can delete at most $k$ vertices in order to obtain a semi-complete digraphs belonging to $\Pi$.

Observe that class $\Pi$ has bounded pathwidth, as by Theorem~\ref{thm:main} and Lemma~\ref{lem:jungle_triple_algo},  a semi-complete digraph of large enough pathwidth contains a sufficiently large triple, in which one of the graph $H_i$ is topologically contained, so also immersed. It follows that if pathwidth of $\Pi$ is bounded by $c_\Pi$, then pathwidth of $\Pi+kv$ is bounded by $c_\Pi+k$. Therefore, we can apply the WIN/WIN approach. Using Theorem~\ref{thm:bundle} in time $g(k,c_\Pi)|V(T)|^3\log |V(T)|$ we either find a $c_\Pi+k+2$-jungle, which is sufficient to provide a negative answer (recall that a $c_\Pi+k+2$-jungle is an obstacle for admitting pathwidth less than $c_\Pi+k+1$), or a path decomposition of width $O((c_\Pi+k)^2)$, on which we can run the algorithm given by Theorem~\ref{thm:MC} applied to formula $\ff$.
\end{proof}

\section{Width parameters of semi-complete digraphs}\label{sec:parameters}

In this section we gather results on the relations between width parameters of semi-complete digraphs.

\subsection{Definitions}

In this subsection we formally define cutwidth and cliquewidth.

Let $D$ be a digraph and let $(v_1,v_2,\ldots,v_n)$ be an ordering of $V(D)$. By \emph{cutwidth} of this ordering we mean the value defined as $\max_{\ell=1,\ldots,n-1} |\{(v_j,v_i)\ |\ (v_j,v_i)\in E(D) \wedge i\leq \ell <j\}$. In other words, for every cut between two consecutive vertices we measure the number of edges directed back, and take maximum of these values. By {\emph{cutwidth}} of the digraph $D$, $\ctw(D)$, we denote the minimum cutwidth among all possible orderings of vertices.

Let $D$ be a digraph and $k$ be a positive integer.
A \emph{$k$-digraph} is a digraph whose vertices are labeled by
integers from $\{1,2,\dots,k\}$. We call the $k$-digraph consisting
of exactly one vertex labeled by some integer from
$\{1,2,\dots,k\}$ an initial $k$-digraph. The \emph{clique-width}
$\cw(D)$ is the smallest integer $k$ such that $G$ can be
constructed by means of repeated application of the following four
operations on $k$-digraphs: {\em $(1)$ introduce}: construction of
an initial $k$-graph labeled by $i$ and denoted by $i(v)$ (that
is, $i(v)$ is a $k$-digraph with $v$ as a single vertex and label
$i$), {\em $(2)$ disjoint union} (denoted by $\oplus$), {\em $(3)$
relabel}: changing all labels $i$ to $j$ (denoted by $\rho_{i\to
j}$), and {\em $(4)$ join}: connecting all vertices labeled by $i$
with all vertices labeled by $j$ by arcs  (denoted by
$\eta_{i,j}$). Using the symbols of these operations, we can
construct well-formed expressions.  An expression is called
\emph{$k$-expression} for $D$ if the digraph produced by performing
these operations, in the order defined by the expression, is
isomorphic to $G$ when labels are removed, and $\cw(D)$ is the
minimum $k$ such that there is a $k$-expression for $D$.

\subsection{Comparison of the parameters}

\begin{lemma}\label{lem:ctw-pw}
For every digraph $D$, $\pw(D)\leq 2\ctw(D)$. Moreover, given an ordering of $V(D)$ of cutwidth $c$, one can in $O(|V(D)|^2)$ time compute a path decomposition of width at most $2c$.
\end{lemma}
\begin{proof}
We provide a method of construction of a path decomposition of width at most $2c$ from an ordering of $V(D)$ of cutwidth $c$.

Let $(v_1,v_2,\ldots,v_n)$ be the ordering of $V(D)$ of cutwidth $c$. Let $F\subseteq E(D)$ be the set of edges $(v_j,v_i)$ such that $j>i$; edges from $F$ will be called {\emph{back edges}}. We construct a path decomposition $W=(W_1,W_2, \dots , W_{n-1},W_n)$ of $D$ by setting
$$W_\ell=\{v_\ell\}\cup \bigcup \{ \{v_i,v_j\}\ |\ i\leq \ell<j \wedge (v_j,v_i)\in E(D)\}.$$
In other words, for each cut between two consecutive vertices in the order (plus one extra at the end of the ordering) we construct a bag that contains endpoints of all the back edges that are cut by this cut plus the last vertex. Observe that $|W_\ell|\leq 2c+1$ for every $1\leq \ell\leq n-1$. It is easy to construct $W$ in $O(|V(D)|^2)$ time using one scan through the ordering $(v_1,v_2,\ldots,v_n)$ and maintaining the number of back edges incident to every vertex at each point. We are left with arguing that $W$ is a path decomposition. Clearly $\bigcup_{i=1}^n W_i=V(D)$.

Consider any vertex $v_\ell$ and an index $j\neq \ell$, such that $v_\ell\in W_j$. Assume first that $j<\ell$. By the definition of $W$, there exists an index $i\leq j$ such that $(v_\ell,v_i)\in E(D)$. Existence of this arc implies that $v_\ell$ has to be contained in every bag between $W_j$ and $W_\ell$. A symmetrical reasoning works also for $j\geq \ell$. We infer that for every vertex $v_\ell$ the set of bags it is contained in form an interval in the path decomposition. The claim that $W_i\cap W_k\subseteq W_j$ for every $1\leq i<j<k\leq n$ follows directly from this observation.

To finish the proof, consider any edge $(v_i,v_j)\in E(G)$. If $i>j$ then $\{v_i,v_j\}\subseteq W_j$, whereas if $i<j$ then $v_i\in W_i$, $v_j\in W_j$ and we are done.
\end{proof}

\begin{lemma}\label{lem:pw-cw}
For every semi-complete digraph $T$, $\cw(T)\leq \pw(T)+2$. Moreover, given a path decomposition of width $p$, one can in $O(|V(T)|^2)$ time compute a $p+2$-expression constructing $T$.
\end{lemma}
\begin{proof}
We provide a method of construction of a $p+2$-expression from a path decomposition of width $p$.

Let $(W_1,W_2,\ldots,W_r)$ be a path decomposition of $T$ of width $p$. By standard means we can assume that the given path decomposition is a {\emph{nice}} path decomposition, i.e., $W_1=W_r=\emptyset$ and $|W_i\setminus W_{i-1}|+|W_{i-1}\setminus W_i|=1$ for every $i=2,3,\ldots,r$. If $W_i=W_{i-1}\cup \{v\}$, we say that $W_i$ {\emph{introduces}} $v$; if $W_{i-1}=W_i\cup\{v\}$, we say that $W_i$ {\emph{forgets}} $v$. Every path decomposition can be easily turned into a nice one of the same width in $O(|V(T)|^2)$ time. Intuitively, between every two bags of the given decomposition we firstly forget the part that needs to be forgotten, and then introduce the part that needs to be introduced. As every vertex is introduced and forgotten exactly once, we have that $r=2|V(T)|$.

We now build a $p+2$-expression that constructs the semi-complete digraph $T$ along the path decomposition. Intuitively, at each step of the construction, every vertex of the bag $W_i$ is assigned a different label between $1$ and $p+1$, while all forgotten vertices are assigned label $p+2$. If we proceed in this manner, we will end up with the whole digraph $T$ labeled with $p+2$, constructed for the last bag. As we begin with an empty graph, we just need to show what to do in the introduce and forget vertex steps.

\medskip
\noindent {\bf{Introduce vertex step.}}

Assume that $W_i=W_{i-1}\cup \{v\}$, i.e., bag $W_i$ introduces vertex $v$. Note that this means that $|W_{i-1}|\leq p$. As labels from $1$ up to $p+1$ are assigned to vertices of $W_{i-1}$ and there are at most $p$ of them, let $q$ be a label that is not assigned. We perform following operations; their correctness is straightforward.
\begin{itemize}
\item perform $\oplus q(v)$: we introduce the new vertex with label $q$;
\item for each $w\in W_{i-1}$ with label $q'$, perform join $\eta_{q,q'}$ if $(v,w)\in E(D)$ and join $\eta_{q',q}$ if $(w,v)\in E(D)$;
\item perform join $\eta_{q,p+2}$ if the new vertex has an outgoing arc to every forgotten vertex.
\end{itemize}

\medskip
\noindent {\bf{Forget vertex step.}}

Assume that $W_i=W_{i-1}\cup \{v\}$, i.e., bag $W_i$ introduces vertex $v$. Let $q\in \{1,2,\ldots,p+1\}$ be the label of $v$. We just perform relabel operation $\rho_{q\to p+2}$, thus moving $v$ to forgotten vertices.
\end{proof}

\section{Monadic Second-Order logic and tournaments}\label{sec:mso}

In this section we recall the definitions of Monadic Second-Order logic and discuss its links to the current work on semi-complete digraphs.

$\MSO_1$ is Monadic Second-Order Logic with quantification over subsets of vertices but not of arcs.
The syntax of $\MSO_1$ of digraphs includes the logical connectives $\vee$, $\land$, $\neg$,
$\Leftrightarrow $,  $\Rightarrow$, variables for
vertices, and sets of vertices, the quantifiers $\forall$, $\exists$ that can be applied
to these variables, and the following three binary relations:
\begin{enumerate}
\item $u\in U$ where $u$ is a vertex variable
and $U$ is a vertex set variable;
\item $\mathbf{A}(u,v)$, where $u,v$ are vertex variables, and the interpretation
is that there exists an arc $(u,v)$ in the graph;
\item equality of variables representing vertices and sets of vertices.
\end{enumerate}

$\MSO_2$ is the natural extension of $\MSO_1$ with quantification also over arcs and subsets of arcs; here, one can additionally check incidency relation between arcs and vertices, i.e., whether a vertex is the head/tail of a given arc. In undirected setting, it is widely known that model checking of formulas of $\MSO_2$ is fixed-parameter tractable, when the parameters are the length of the formula and the treewidth of the graph. As far as $\MSO_1$ is concerned, model checking of formulas of $\MSO_1$ is fixed-parameter tractable, when the parameters are the length of the formula and the cliquewidth of the graph. These results in fact hold not only for undirected graphs, but for structures with binary relations in general, in particular for digraphs. The following result follows from the work of Courcelle, Makowsky and Rotics~\cite{CourcelleMR00}; we remark that the original paper treats of undirected graphs, but in fact the results hold also in the directed setting (cf.~\cite{dam-GanianH10,GanianHO11,abs-0709-1433}).

\begin{theorem}
There exist an algorithm with running time $f(||\ff||,k)|V(T)|^2$ that given an $\MSO_1$ formula $\ff$ checks whether $\ff$ is satisfied in a semi-complete digraph $T$, given together with a $k$-expression constructing it.
\end{theorem}

Lemma~\ref{lem:pw-cw} assert that cliquewidth of a tournament is bounded its pathwidth plus $2$. Moreover, the proof gives explicit construction of the corresponding expression. Hence, Theorem~\ref{thm:MC} follows as an immediate corollary. We note that as Lemmas~\ref{lem:ctw-pw} and~\ref{lem:pw-cw} show essentially the same for cutwidth, the analogous result holds also for cutwidth.

It is tempting to conjecture that the tractability for $\MSO_1$ and pathwidth or cutwidth could be extended also to $\MSO_2$, as the decompositions resemble treewidth in the undirected setting. We use the following lemma to shows that this is unfortunately not true.

\begin{lemma}\label{thm:mso2-niet}
There exists a constant-size $\MSO_2$ formula $\psi$ over a signature enriched with three unary relations on vertices, such that checking whether $\psi$ is satisfied in a transitive tournament is \NP-hard.
\end{lemma}

Before we prove the lemma, let us shortly deliberate on how Theorem~\ref{thm:mso-main} follows from it. Observe that one can replace each vertex with a constant size subdigraph that encodes satisfaction of unary relations. In $\psi$ we then replace the unary relations with constant size tests, thus obtaining a constant size formula $\psi'$, whose model checking on semi-complete digraphs of constant pathwidth and cutwidth is \NP-hard.

\begin{proof}[Proof of Lemma~\ref{thm:mso2-niet}]
We provide a polynomial-time reduction from the \tcnfsat problem. We are given a boolean formula $\ff$ with $n$ variables $x_1,x_2,\ldots,x_n$ and $m$ clauses $C_1,C_2,\ldots,C_m$. We are to construct a transitive tournament $T$ with three unary relations on vertices such that $\ff$ is satisfiable if and only if $\psi$, the constant size $\MSO_2$ formula that will be constructed while describing the reduction, is true in $T$. Intuitively, the unary relations in $T$ will encode the whole formula $\ff$, while $\psi$ is simply an $\MSO_2$-definable check that nondeterministically guesses the evaluation of variables and checks it.

We will use three unary relations on vertices, denoted $P$, $B$ and $C$. The tournament consists of $m(2n+1)+1$ vertices $v_0,v_1,\ldots,v_{m(2n+1)}$, where the edge set is defined as $E(T)=\{(v_i,v_j)\ |\ i>j\}$. We define $B$ (border) to be true in $v_i$ if and only if $2n+1|i$. Thus, the whole tournament is divided into $m$ intervals between consecutive vertices satisfying $B$, each of size $2n$. Each of these intervals will be responsible for checking one of the clauses. The $2n$ vertices in each interval correspond to literals of variables of $\ff$. We define $P$ (parity) to be satisfied in every second vertex of each interval, so that $P$ is satisfied in the first vertex of each interval. The first pair of vertices corresponds to literals $x_1,\neg x_1$, the second to $x_2,\neg x_2$ etc. In the $i$-th interval we make the $C$ (check) relation true in vertices corresponding to literals appearing in the clause $C_i$. This concludes the construction.

We now build the formula $\psi$ that checks existence of an assignment satisfying $\ff$. The correctness of the reduction will follow directly from the construction of $\psi$.

Firstly, we quantify existentially over a subset of edges $M$ and subset of vertices $X$. $X$ will be a set of all vertices corresponding to literals that are true in the assignment. We will construct $\psi$ in such a manner that $M$ will be exactly the set of arcs $\{(v_{i+2n+1},v_i)\ |\ 0\leq i\leq(m-1)(2n+1)\}$. We use $M$ to transfer information on the assignment between consecutive intervals.

Formally, in $\psi$ we express following properties of $X$ and $M$ (their expressibility in $\MSO_2$ by formulas of constant size is straightforward):
\begin{enumerate}
\item[(1)] $(v_{2n+1},v_0)\in M$ (note that $v_0$ and $v_{2n+1}$ can be defined as the first and second vertex satisfying $B$).
\item[(2)] For every two pairs of consecutive vertices $v_i,v_{i+1}$ and $v_j,v_{j+1}$, if $(v_j,v_i)\in M$ then $(v_{j+1},v_{i+1})\in M$.
\item[(3)] For every subset $N\subseteq M$, if $N$ satisfies (1) and (2) then $N=M$.
\item[(4)] Vertices satisfying $B$ are not in $X$.
\item[(5)] For every two consecutive vertices $v_i,v_{i+1}$, such that $v_i,v_{i+1}\notin B$, $v_i\in P$ and $v_{i+1}\notin P$, exactly one of the vertices $v_i,v_{i+1}$ belongs to $X$ (exactly one literal of every variable is true).
\item[(6)] For every $(v_i,v_j)\in M$, $v_i\in X$ if and only if $v_j\in X$.
\item[(7)] For every interval between two subsequent vertices satisfying $B$, at least one of the vertices satisfying $C$ belongs to $X$.
\end{enumerate}
Properties (1), (2) and (3) assert that $M=\{(v_{i+2n+1},v_i)\ |\ 0\leq i\leq(m-1)(2n+1)\}$. Properties (4) and (5) assert that in each interval $X$ corresponds to some valid assignment, while property (6) asserts that the assignments in all the intervals are equal. Property (7) checks whether each of the clauses is satisfied.
\end{proof}

\section{Irrelevant vertex in a triple for \himmersion}\label{sec:irrelevant}
In this section we show how to identify a vertex that is irrelevant for \himmersion problem, in a semi-complete graph containing a sufficiently large triple.
Let $p(x)= 80x^2+80x+5$. We prove the following lemma.

\begin{lemma}[Lemma~\ref{thm:irrelevant}, restated]\label{thm:irrelevant-app}
Let $\Ii=((H;u_1,u_2,\ldots,u_t),(T;v_1,v_2,\ldots,v_t))$ be an instance of the \himmersion{} problem, where $T$ is a semi-complete graph containing a $p(k)$-triple $(A,B,C)$ disjoint with $\{v_1,v_2,\ldots,v_t\}$ and $k=|H|$. Then it is possible in time $O(p(k)^2|V(T)|^2)$ to identify a vertex $x\in B$ such that $\Ii$ is a YES instance of \himmersion{} if and only if $\Ii'=((H;u_1,u_2,\ldots,u_t),(T\setminus\{x\};v_1,v_2,\ldots,v_t))$ is a YES instance. 
\end{lemma}

Before we proceed to the proof of Lemma~\ref{thm:irrelevant-app}, we need to make several auxiliary observations.

Let $\eta$ be a solution to \himmersion{} instance and let $\Pp$ be the family of paths being images of all the edges in $H$, i.e., $\Pp=\eta(E(H))$.
We call an edge (a vertex) {\em{used by $\Pp$}}, if it belongs to some path from $\Pp$. 
We omit the family $\Pp$ whenever it is clear from the context. An edge (a vertex) which is not used is called a {\em{free}} edge (vertex).

\begin{observation}\label{lem:bound-deg}
If $\Pp$ is a family of paths containing simple paths only, then every vertex in $T$ is adjacent to at most $k$ used ingoing edges and at most $k$ used outgoing edges.
\end{observation}
\begin{proof}
Otherwise, there is a path in the solution that visits that vertex at least two times. Therefore, there is a cycle in this path, which contradicts its simplicity.
\end{proof}

Let $\Qq=\eta(E(H))$ be the family of paths, where $\eta$ is the solution to \himmersion instance $\Ii$ that minimizes the total sum of paths lengths. Firstly, we observe some easy properties of $\Qq$.

\begin{observation}\label{lem:bound-match}
Every path from $\Qq$ uses at most $2$ edges from the matching between $C$ and $A$.
\end{observation}
\begin{proof}
Assume otherwise, that there is a path $P\in \Qq$ that uses three edges of the matching: $(c_1,a_1)$, $(c_2,a_2)$, $(c_3,a_3)$, appearing in this order on the path. By Observation~\ref{lem:bound-deg}, for at most $k$ vertices $v\in B$ the edge $(a_1,v)$ is used. For the same reason, for at most $k$ vertices $v\in B$ the edge $(v,c_3)$ is used. As $|B|>2k$, there exists $v\in B$ such that $(a_1,v)$ and $(v,c_3)$ are not used. Now replace the part of $P$ appearing between $a_1$ and $c_3$ with $a_1\to v\to c_3$. We obtain an immersion with smaller sum of lengths of the paths, a contradiction.
\end{proof}

\begin{observation}\label{lem:boundA}
Every path from $\Qq$ uses at most $2k+4$ vertices from $A$.
\end{observation}
\begin{proof}
Assume otherwise, that there is a path $P\in \Qq$  passing  through at least $2k+5$ vertices from $A$. By Observation~\ref{lem:bound-match}, at most $2k$ of them are endpoints of used edges of the matching between $C$ and $A$. Therefore, there are at least $5$ visited vertices, which are endpoints of an unused edge of the matching. Let us denote any $5$ of them by $a_1,a_2,a_3,a_4,a_5$ and assume that they appear on $P$ in this order. Let $(c_5,a_5)$ be the edge of the matching between $C$ and $A$ that is adjacent to $a_5$. By the same reasoning as in the proof of Observation~\ref{lem:bound-match}, there exists a vertex $v\in B$, such that $(a_1,v)$ and $(v,c_5)$ are unused edges. Substitute the part of the path $P$ between $a_1$ and $a_5$ by the path $a_1\to v\to c_5\to a_5$, which consists only of unused edges. We obtain an expansion with smaller sum of lengths of the paths, a contradiction. 
\end{proof}

A symmetrical reasoning yields the following observation.

\begin{observation}\label{lem:boundC}
Every path from $\Qq$ uses at most $2k+4$ vertices from $C$.
\end{observation}

Finally, we prove a similar property for $B$.

\begin{observation}\label{lem:boundB}
Every path from $\Qq$ uses at most $4$ vertices from $B$.
\end{observation}
\begin{proof}
Assume otherwise, that there is a path $P\in \Qq$ such that it passes through at least $5$ vertices from $B$. Let us denote any $5$ of them by $b_1,b_2,b_3,b_4,b_5$ and assume that they appear on $P$ in this order. By Observation~\ref{lem:bound-deg} there are at most $k$ ingoing edges adjacent to $b_1$ used, and there are at most $k$ outgoing edges adjacent to $b_5$ used. Moreover, by Observation~\ref{lem:bound-match} there are at most $2k$ edges of the matching between $C$ and $A$ used. As $p(k)>4k$, we conclude that there is an unused edge of the matching $(c,a)$, such that edges $(b_1,c)$ and $(a,b_5)$ are also unused. Substitute the part of the path $P$ between $b_1$ and $b_5$ with the path $b_1\to c\to a\to b_5$. We obtain an immersion with smaller sum of lengths of the paths, a contradiction.
\end{proof}

From Observations \ref{lem:boundA}-\ref{lem:boundB} we obtain the following corollary.

\begin{corollary}\label{lem:free}
In the set $B$ there are at least $5k$ vertices free from $\Qq$. Moreover, within the matching between $C$ and $A$ there are at least $4k$ edges having both endpoints free from $\Qq$.
\end{corollary}

We note that the Corollary~\ref{lem:free} holds also for much larger values than $5k$, $4k$, respectively; we choose to state it in this way to show how many free vertices from $B$ and free edges of the matching we actually use in the proof of Lemma~\ref{thm:irrelevant-app}. We need one more auxiliary lemma that will prove itself useful.

\begin{lemma}\label{lem:bitour}
Let $T=(V_1\cup V_2,E)$ be a semi-complete bipartite graph, i.e., a directed graph, where edges are only between $V_1$ and $V_2$, however, for every $v_1\in V_1$ and $v_2\in V_2$ at least one of the edges $(v_1,v_2)$ and $(v_2,v_1)$ is present. Then at least one of the assertions is satisfied:
\begin{itemize}
\item[(a)] for every $v_1\in V_1$ there exists $v_2\in V_2$ such that $(v_1,v_2)\in E$;
\item[(b)] for every $v_2\in V_2$ there exists $v_1\in V_1$ such that $(v_2,v_1)\in E$.
\end{itemize}
\end{lemma}
\begin{proof}
Assume that (a) does not hold. That means that there is some $v_0\in V_1$ such that for all $v_2\in V_2$ we have $(v_2,v_0)\in E$. Then we can always pick $v_0$ as $v_1$ in the statement of (b), so (b) holds.
\end{proof}

Observe that by reversing all the edges we can obtain a symmetrical lemma, where we assert existence of in-neighbours instead of out-neighbours.

We are now ready to prove Lemma~\ref{thm:irrelevant-app}. Whenever we will refer to the {\em{matching}}, we mean the matching between $C$ and $A$.

\begin{proof}[Proof of Lemma~\ref{thm:irrelevant-app}]
We give an algorithm that outputs a vertex $x\in B$, such that if there exists a solution to the the given instance, then there exists also a solution in which no path passes through $x$. The algorithm will run in time $O(p(k)^2|V(T)|^2)$.

We proceed in three steps. 
The first step is to identify in $O(p(k)^2|V(T)|^2)$ time a set $X\subseteq B$, $|X|\geq 16k^2+16k+1$, such that if $\Ii$ is a YES instance, then 
for every $x\in X$ there is a solution $\eta$ with $\Pp=\eta(E(H))$ having following properties:
\begin{itemize}
\item at least $3k$ vertices of $B$ are free from $\Pp$;
\item at least $2k$ edges of the matching have both endpoints free from $\Pp$;
\item if $x$ is accessed by some path $P\in \Pp$ from a vertex $v$, then $v\in A$.
\end{itemize}

The second step of the proof is to show that one can identify in $O(p(k)^2|V(T)|^2)$ time a
vertex $x\in X$ such that if $\Ii$ is a yes instance, then there is a solution with $\Pp=\eta(E(H))$ having following properties:
 \begin{itemize}
\item at least $k$ vertices of $B$ are free from $\Pp$;
\item if $x$ is accessed by some path $P\in \Pp$ from a vertex $v$, then $v\in A$;
\item if $x$ is left by some path $P\in \Pp$ to a vertex $v$, then $v\in C$.
 \end{itemize}

The final, concluding step the proof of the lemma, is to show that there is a solution $\Pp=\eta(E(H))$ such that
\begin{itemize}
\item No path from $\Pp$ is using $x$.
\end{itemize}
 
\medskip
We proceed with the first step. 
Let $\Qq=\eta(E(H))$, where $\eta$ is the solution for the \himmersion{} instance with the minimum sum of lengths of the paths. 

For every vertex $b\in B$, we identify two sets: $G_b$, $R_b$. The set $R_b$ consists of those in-neighbours of $b$ outside $A$, which are in-neighbours of at least than $6k$ vertices from $B$, while $G_b$ consists of the rest of in-neighbours of $b$ outside $A$. Formally,
\begin{eqnarray*}
R_b=\{v\in V(T)\setminus A: (v,b)\in E \wedge |N^{+}(v)\cap B|\geq 6k\},\\
G_b=\{v\in V(T)\setminus A: (v,b)\in E \wedge |N^{+}(v)\cap B|< 6k\}. 
\end{eqnarray*}
Note that $R_b,G_b$ can be distinguished for all $b$ in $O(p(k)^2|V(T)|)$ time. Let $B_\emptyset$ be the set of those vertices $b\in B$, for which $G_b=\emptyset$. We claim that if $|B_\emptyset|\geq 16k^2+16k+1$, then we can set $X=B_\emptyset$.

Take any $b\in B_\emptyset$. We argue that we can reroute the paths of $\Qq$ that access $b$ from outside $A$ in such a manner, that during rerouting each of them we use at most one additional free vertex from $B$ and at most one additional edge from the matching. We reroute the paths one by one. Take a path $P$ that accesses $b$ from outside $A$, and let $v$ be the previous vertex on the path. As $G_b=\emptyset$, $v\in R_b$. Therefore, $v$ has at least $6k$ out-neighbours in $B$. Out of them, at most $4k$ are not free with respect to $\Qq$, due to Observation~\ref{lem:boundB}, while at most $k-1$ were used by previous reroutings. Therefore, there is a vertex $b'\in B\cap N^{+}(v)$, such that $b'$ is still free. Thus we can substitute usage of the edge $(v,b)$ on $P$ with the path $v\to b'\to c\to a\to b$, where $(c,a)$ is an arbitrary edge of the matching that still has both endpoints free, which exists due to using at most $k-1$ of them so far.

We are now left with the case, when $|B_\emptyset|<16k^2+16k+1$.  Let $B_g=B\setminus B_\emptyset$. Then $|B_g|\geq 4(16k^2+16k+1)$. We construct a semi-complete digraph $S=(B_g,L)$ as follows. For every $b_1,b_2\in B_g$, $b_1\neq b_2$, we put arc an $(b_1, b_2)$ if for every $v\in  G_{b_1}$, either $v\in  G_{b_1}\cap G_{b_2}$, or $v$ has an out-neighbour in $G_{b_2}$. Similarly, we put arc $(b_2, b_1)$ if for every $v\in  G_{b_2}$, either $v\in  G_{b_1}\cap G_{b_2}$, or $v$ has an out-neighbour in $G_{b_1}$. By Lemma~\ref{lem:bitour}, for every pair of distinct $b_1,b_2\in B_g$, there is at least one arc with endpoints $b_1$ and $b_2$, hence $S$ is semi-complete. The definition of $S$ gives raise to a straightforward algorithm constructing it in $O(p(k)^2|V(T)|^2)$ time.

\begin{figure}[htp]
\begin{center}
\includegraphics[width=10cm]{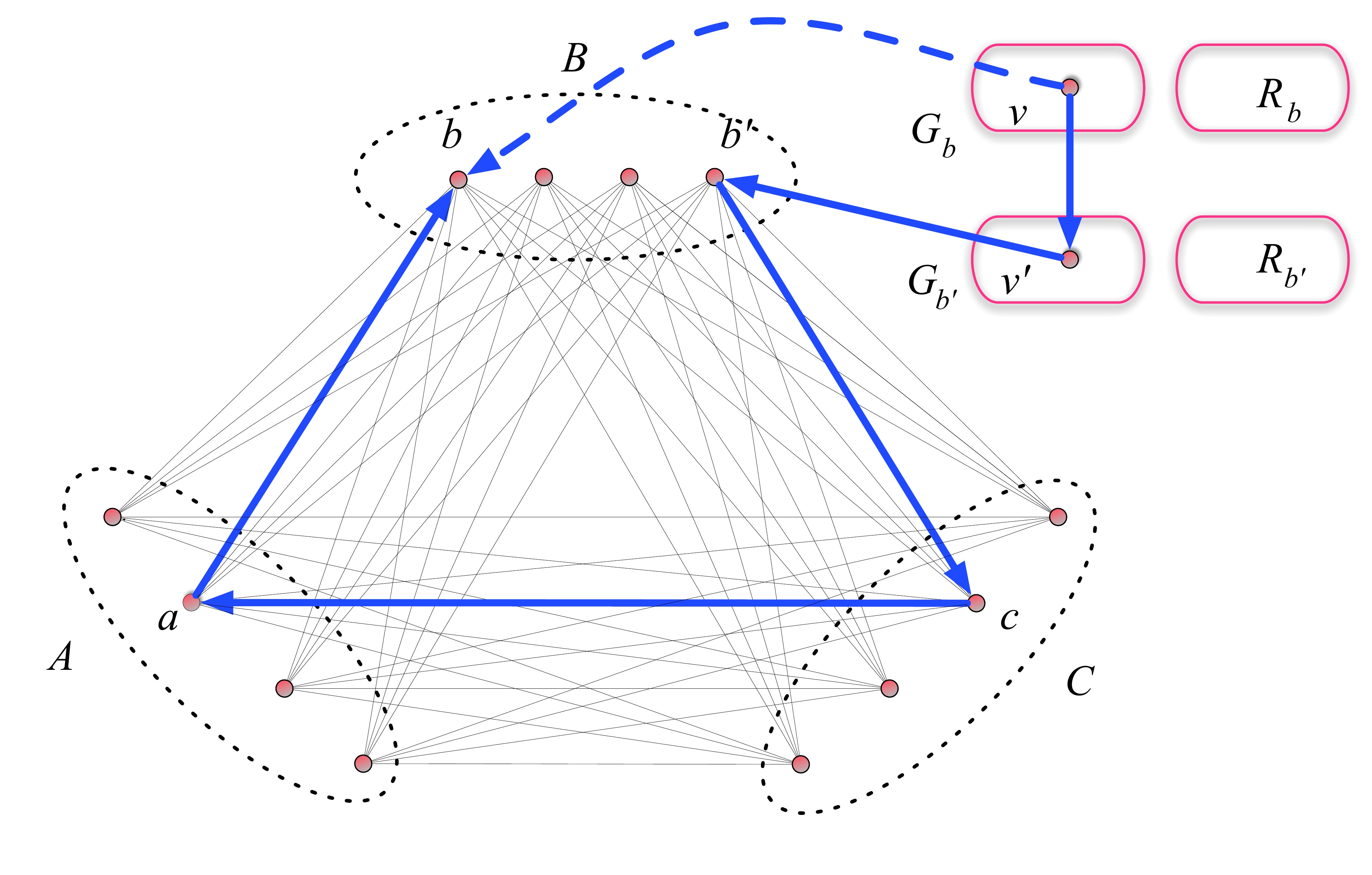}
\caption{Rerouting strategy for a path accessing vertex $b$ from $G_b$.}\label{fig:irrelevant}
\end{center}
\end{figure}

Let $X$ be the set of vertices of $B_g$ that have out-degree at least $6k^2+6k$ in $S$; note that $X$ can be constructed in $O(p(k)^2)$ time. Observe that $|X|\geq 16k^2+16k+1$, as otherwise the sum of the out-degrees in $S$ would be at most $(16k^2+16k)(|B_g|-1) + (|B_g|-16k^2-16k)(6k^2+6k-1)$, which is smaller than $\binom{|B_g|}{2}$ for $|B_g|\geq 4(16k^2+16k+1)$. We now claim that for every $b\in X$, every path of $\Qq$ using vertex $b$ can be rerouted at the cost of using at most two free vertices of $B$ and at most two edges from the matching that have still both endpoints free. We perform reroutings one by one. Assume that there is a path $P\in \Qq$ accessing $b$ from outside $A$. Let $v$ be the predecessor of $b$ on $P$. If $v\in R_b$, then we use the same rerouting strategy as in the case of large $B_\emptyset$. Assume then that $v\in G_b$. As $b\in X$, its out-degree in $S$ is at least  $6k^2+6k$. This means that there are at  least $6k^2+6k$ vertices $b'$ in $B_g$ and corresponding vertices $v_{b'}\in N^{-}(b')$, such that for every $b'$ either $v_{b'}=v$ or $(v,v_{b'})\in E$. Out of these $6k^2+6k$ vertices $b'$, at most $4k$ are not free due to Observation~\ref{lem:boundB}, at most $2k-2$ were used in previous reroutings, which leaves us with at least $6k^2$ vertices $b'$ still being free. If for any such $b'$ we have $v_{b'}=v$, we follow the same rerouting strategy as in the case of large $B_\emptyset$. Assume then that these $6k^2$ vertices $v_{b'}$ are all distinct from $v$; note that, however, they are not necessarily distinct from each other. As each $v_{b'}$ belongs to $G_{b'}$, $v_{b'}$ can have at most $6k-1$ out-neighbours in $B$. Hence, each vertex of $V(T)$ can occur among these $6k^2$ vertices $v_{b'}$ at most $6k-1$ times, so we can distinguish at least $k+1$ pairwise distinct vertices $v_{b'}$. We have that edges $(v,v_{b'})$ and $(v_{b'},b')$ exist, while $b'$ is still free. By Observation~\ref{lem:bound-deg}, at most $k$ edges $(v,v_{b'})$ are used by some paths, which leaves us at least one $v_{b'}$, for which edge $(v,v_{b'})$ is free. We can now substitute the edge $(v,b)$ in $P$ by the path $v\to v_{b'}\to b'\to c\to a\to v$, where $(c,a)$ is an arbitrarily chosen edge from the matching that still has both endpoints free that was not yet used, which exists due to using at most $2k-2$ of them so far. 
See Fig.~\ref{fig:irrelevant}.
After rerouting, remove all loops that could be created on the paths, so that Observation~\ref{lem:bound-deg} still holds. Observe that in this manner we use additional vertex $b'$ that was free, additional one edge $(c,a)$ from the matching, whereas passing the path through $v_{b'}$ can spoil at most one edge of the matching that still has both endpoints free, and at most one free vertex from $B$. This concludes the construction of the set $X$.

\medskip We proceed with the second step of the proof. 
We mimic the rerouting arguments from the first step to obtain a vertex $x\in X$ with the following property: the rerouted family of paths $\Pp$ obtained from previous arguments that can access $x$ only from $A$, can be further rerouted so that every path can only leave $x$ by accessing some vertex from $C$.

For every $b\in X$ consider sets $R_b'$ and $G_b'$ defined similarly as before (again, they can be computed in $O(p(k)^2|V(T)|)$ time):
\begin{eqnarray*}
R_b'=\{v\in V(T)\setminus C: (b,v)\in E \wedge |N^{-}(v)\cap B|\geq 8k\},  \\
G_b'=\{v\in V(T)\setminus C: (b,v)\in E \wedge |N^{-}(v)\cap B|< 8k\} 
\end{eqnarray*}
Assume first that there is some $y\in X$, such that $G_y'=\emptyset$. We argue that we can in such a case set $x=y$. Firstly, reroute the solution that minimizes the total sum of lengths of the paths obtaining a solution with the family of paths $\Pp$ that uses at most $2k$ additional free vertices from $B$ and at most $2k$ additional edges free edges from the matching, but does not access $y$ from outside $A$. One by one we reroute paths using $y$. Each rerouting will cost at most one free vertex from $B$ and at most one edge from the matching that has still both endpoints free. Let $P$ be a path from the solution that passes through $y$ and let $v\in R_y'$ be the next vertex on $P$. The vertex $v$ has at least $8k$ in-neighbours in $B$; at most $4k$ of them could be used by the original solution, at most $2k$ of them could be used in rerouting during the first phase and at most $k-1$ of them could be used during previous reroutings in this phase. Therefore, we are left with at least one vertex $y'\in B$ that is still free, such that $(y',v)\in E(T)$. We can now substitute the edge $(y,v)$ in $P$ by the path $y\to c\to a\to y'\to v$, where $(c,a)$ is an arbitrarily chosen edge from the matching that was not yet used, which exists due to using at most $3k-1$ of them so far.

We are left with the case, when $G_y'$ is nonempty for all $y\in X$. Construct a digraph $S'=(X,L')$ in a symmetrical manner to the previous construction: we put an edge $(b_1,b_2)\in L'$ iff for every $v_{b_2}\in G_{b_2}'$ there exists $v_{b_1}\in G_{b_1}'$ such that $v_{b_1}=v_{b_2}$ or $(v_{b_1},v_{b_2})\in E$. The remark after Lemma~\ref{lem:bitour} ensures that $S'$ is semi-complete. Again, $S'$ can be computed in $O(p(k)^2|V(T)|^2)$ time.

As $|X|\geq 16k^2+16k+1$, there exists $x\in X$, which has in-degree at least $8k^2+8k$ in $S'$; note that $x$ can be found in $O(p(k)^2)$ time. As before, we argue that after the first rerouting phase for $x$, we can additionally reroute the paths so that every path can leave $x$ only into $C$. We reroute the paths one by one; each rerouting uses at most two vertices free vertices from $B$ and at most two edges from the matching. As the in-degree of $x$ in $S'$ is at least $8k^2+8k$, we have at least $8k^2+8k$ vertices $x'\in X$ and corresponding $v_{x'}\in G_{x'}'$, such that $v_{x'}=v$ or $(v_{x'},v)\in E$. At most $4k$ of them were used in $\Qq$, at most $2k$ were used in previous phase of rerouting, and at most $2k-2$ of them were used in this phase of rerouting. This leaves at least $8k^2$ vertices $x'$ which are still free. If for any of them we have $v_{x'}=v$, we can make the rerouting similarly as in the previous case: we substitute the edge $(x,v)$ with the path $x\to c\to a\to x'\to v$, where $(c,a)$ is an arbitrary edge of the matching that still has both endpoints free, which exists due to using at most $4k-2$ of them so far. Assume then, that all vertices $v_{x'}$ are distinct from $v$; note that, however, they are not necessarily distinct from each other. As every for every $x'$ we have $v_{x'}\in G_{x'}'$, by the definition of $G_{x'}'$ the vertices $v_{x'}$ can have at most $8k-1$ in-neighbours in $X$. This means that every vertex of $V(T)$ can occur among vertices $v_{x'}$ at most $8k-1$ times, which proves that there are at least $k+1$ pairwise distinct vertices among them. By Observation~\ref{lem:bound-deg}, for at most $k$ of them the edge connecting to $v$ can be already used, which leaves us with a single vertex $x'$, such that edges $(x',v_{x'})$ and $(v_{x'},v)$ exist and are not yet used, whereas $x'$ is still free. Now we can perform the rerouting as follows: we substitute the edge $(x,v)$ in $P$ by the path $x\to c\to a\to x'\to v_{x'}\to v$, where $(c,a)$ is an arbitrary edge from the matching that was still has both endpoints free, which exists due to at most $4k-2$ of them so far. Similarly as before, we use one additional free vertex $x'$ from $B$, one additional edge $(c,a)$ from the matching, whereas usage of $v_{x'}$ can spoil at most one free vertex from $B$ and at most one edge from the matching. After this, we delete all the possible loops created on the paths in order to make Observation~\ref{lem:bound-deg} still hold. This concludes the construction of the vertex $x$.

\medskip
To finish the proof, we need to show that after performing two phases of rerouting obtaining a solution $\eta$, whose paths can access and leave $x$ only from $A$ and into $C$, we can reroute every path so that it does not pass through $x$ at all. Note that so far we have used at most $4k$ vertices from $B$, so we still have at least $k$ vertices unused. 

If $x\in \eta(u)$ for some $u\in V(H)$, then we can simply move the image $u$: we consider $\eta'$ that differs from $\eta$ by substituting $x$ for $x'$, where $x'$ is any free vertex from $B$ that is not in $\eta(V(H))$ ($k\geq |V(H)|$, so such a vertex exists). Note that we can make this move as $x$ is not a root vertex: the triple does not contain any root vertices. In case when $x\notin \eta(V(H))$, we can perform the same rerouting scheme for any $x'$ that is still free in $B$. Now, no path uses $x$, so $\Ii'=((H,u_1,u_2,\ldots,u_t),(T\setminus\{x\},v_1,v_2,\ldots,v_t))$ is a YES instance if $\Ii$ was.
\end{proof}

\section{Dynamic programming routine for rooted immersion}\label{sec:app-DP}

In this subsection we present a dynamic programming routine that for rooted digraph $(H;v_1,v_2,\ldots,v_t)$, rooted semi-complete digraph $(T;w_1,w_2,\ldots,w_t)$, and path decomposition of $T$ computes whether there exists an immersion $\eta$ of $H$ into $T$ such that $\eta(v_i)=w_i$ for all $i\in\{1,2,\ldots,t\}$. First, we explain the terminology used to describe the constructed parts of immersion model. Then, we proceed to the dynamic programming routine itself. But before all of these, let us give some intuition behind the algorithm.

The main idea of our routine is as follows: we will encode the interaction of the model of $H$ with all the already introduced vertices as sequences of paths. Every such path has to end in the separator, but can begin in any forgotten vertex as it may be accessed from a vertex not yet introduced. In general setting the dynamic program would need to remember this very vertex in order to check whether it can be indeed accessed; that would yield an \XP algorithm and, in essence, this is exactly the idea behind the Fradkin and Seymour algorithm. However, if the digraph is semi-complete, then between every not yet introduced vertex and every forgotten vertex there is an edge. Therefore, we do not need to remember the forgotten vertex itself to check accessibility; a marker saying {\emph{forgotten}}, together with information about which markers in fact represent the same vertices, will suffice.

We hope that a reader well-familiar with construction of dynamic programs on various decompositions already has a crude idea about how the computation will be performed. Let us proceed with the details in the next subsections.

\subsection{Terminology}

In the definitions we use two special symbols: $\FF,\UU$; the reader can think of them as an arbitrary element of $A\setminus B$ ({\emph{forgotten}}) and $B\setminus A$ ({\emph{unknown}}), respectively. Let $\iota:V(T)\to (A\cap B)\cup \{\FF,\UU\}$ be defined as follows: $\iota(v)=v$ if $v\in A\cap B$, whereas $\iota(v)=\FF$ for $v\in A\setminus B$ and $\iota(v)=\UU$ for $v\in B\setminus A$.

\begin{definition} Let $P$ be  a path.
A sequence of paths $(P_1,P_2,\ldots,P_h)$ is a {\emph{\ptrace}} of $P$ with respect to $(A,B)$, if $P_i$, $i\in\{1,2,\ldots,h\}$ are all maximal subpaths of $P$ that are fully contained in $A$, while the indices in the sequence reflect their ordering on the path $P$.

Let $(P_1,P_2,\ldots,P_h)$ be the \ptrace{} of $P$ with respect to $(A,B)$. An {\emph{signature}} of $P$ on $(A,B)$ is a sequence of pairs $((b_1,e_1),(b_2,e_2),\ldots,(b_h,e_h))$, where $b_h,e_h\in (A\cap B)\cup\{\FF\}$, such that for every $i\in\{1,2,\ldots,h\}$:
\begin{itemize}
\item $b_i$ is the beginning of path $P_i$ if $b_i\in A\cap B$, and $\FF$ otherwise;
\item $e_i$ is the end of path $P_i$ if $e_i\in A\cap B$, and $\FF$ otherwise.
\end{itemize}
\end{definition}

In other words, $b_i,e_i$ are images of the beginning and the end of the path $P_i$ in the mapping $\iota$. Observe following properties of the introduced notion:
\begin{itemize}
\item Signature of path $P$ on separation $(A,B)$ depends only on its \ptrace; therefore, we can also consider signatures of \ptraces.
\item It can happen that $b_i=e_i\neq \FF$ only if $P_i$ consists of only one vertex $b_i=e_i$.
\item From the definition of separation it follows that only for $i=h$ it can happen that $e_i=\FF$, as there is no arc from $A\setminus B$ to $B\setminus A$.
\item The empty signature
corresponds to $P$ entirely contained in $B\setminus A$.
\item The signature consisting of one term $(\FF,\FF)$ corresponds to $P$ entirely contained in $A\setminus B$.
\end{itemize}

Finally, we are able to encode relevant information about a given model of $H$.

\begin{definition}
Let $\eta$ be an immersion of a digraph $H$ in $T$. An {\emph{immersion signature}} of $\eta$ on $(A,B)$ is a mapping $\rho$ together with equivalence relation $\equiv$ on the set of all the pairs of form $(\FF,e)$ appearing in the image of $\rho$, such that:
\begin{itemize}
\item for every $v\in V(H)$, $\rho(v)=\iota(\eta(v))$;
\item for every $e\in E(H)$, $\rho(e)$ is a signature of $\eta(e)$.
\item $(\FF,e_1)\equiv (\FF,e_1)$ if and only if markers $\FF$ in both pairs correspond to the same forgotten vertex before being mapped by $\iota$.
\end{itemize}
\end{definition}

We remark that the same pair of form $(\FF,e)$ can appear in different signatures; in this case we consider all the appearances as different pairs. The sole purpose of introducing the relation $\equiv$ is forbidding reusage of edges jumping into the forgotten region. We denote the set of possible immersion signatures on separation $(A,B)$ by $\ef{A}{B}$.

We now show that the number of possible immersion signatures is bounded, providing the separation is of small order.

\begin{lemma}\label{lem:etrace-bound}
If $|V(H)|=k, |E(H)|=\ell, |A\cap B|=m$, then the number of possible different immersion signatures on $(A,B)$ is bounded by
$$(m+2)^k \cdot ((m+2)^m\cdot m!\cdot (m+2))^\ell\cdot B_{(m+1)\ell}=2^{O(k\log m+m\ell(\log \ell+\log m))}.$$
Moreover, all of them can be enumerated in $2^{O(k\log m+m\ell(\log \ell+\log m))}$ time.
\end{lemma}
\begin{proof}
The consecutive terms correspond to:
\begin{enumerate}
\item the choice of mapping $\rho$ on $V(H)$;
\item for every edge $e\in E(H)$ the complete information about the signature $\rho(e)$:
\begin{itemize}
\item for every element of $A\cap B$, whether it will be the end of some path in the signature, and in this case, the value of corresponding beginning (a vertex from $A\cap B$ or $\FF$),
\item the ordering of pairs along the signature,
\item and whether to append a pair of form $(b,\FF)$ at the end of the signature $\rho(e)$, and in this case, the value of $b$ (a vertex from $A\cap B$ or $\FF$);
\end{itemize}
\item partition of at most $(m+1)l$ pairs in all the signatures from $\rho(E(H))$ into equivalence classes with respect to $\equiv$.
\end{enumerate}
In the last term we made use of Bell numbers $B_n$, for which a trivial bound $B_n\leq n^n$ applies.

It is easy to check that using all these information one can reconstruct the whole signature. For every object constructed in the manner above we can check in time polynomial in $k,l,m$, whether it corresponds to a possible signature. This yields the enumeration algorithm.
\end{proof}

\subsection{Dynamic programming routine}

We are now ready to prove the main theorem of this section.

\begin{theorem}\label{thm:dpe}
There exists an algorithm that given a distinctly rooted digraph $(H;v_1,v_2,\ldots,v_t)$ with $k$ vertices and $\ell$ edges and a semi-complete distinctly rooted digraph $(T;w_1,w_2,\ldots,w_t)$ together with its path decomposition of width $p$, checks in time $2^{O(k\log p + p\ell(\log \ell+\log p))}|V(T)|$, whether $H$ can be immersed into $T$ while preserving roots.
\end{theorem}
\begin{proof}
Let $W = (W_1, \dots , W_r)$ be a path decomposition of $T$ of width $p$.
We can assume that  $W$ is a \emph{nice} decomposition, i.e.  $W_1=W_r=\emptyset$ and for every $i\in \{2,  \dots, r\}$, $|W_i \setminus W_{i-1}|+|W_{i-1} \setminus W_{i}|=1$. If $W_i=W_{i-1}\cup\{v\}$, we say that $W_i$ {\emph{introduces}} $v$; if $W_{i-1}=W_i\cup\{v\}$, we say that $W_i$ {\emph{forgets}} $v$.
Via standard arguments, it is easy to check that every path decomposition can be adjusted to be nice in linear time. Intuitively, between every two bags of the given decomposition we firstly forget the part that needs to be forgotten, and then introduce the part that needs to be introduced. We also note remark that for every   $i\in \{1,  \dots, r\}$, the sets $A=W_1\cup \cdots \cup W_{i}$ and
 $B=W_i\cup \cdots \cup W_{r}$ form a separation of $T$ with separator $W_i$.

The number of possible signatures on any separation $(A,B)$ from the decomposition is bounded by $2^{O(k\log p + pl(\log \ell+\log p))}$ by Lemma~\ref{lem:etrace-bound}. We show how to compute the values of a binary table $D_{(A,B)}:\ef{A}{B}\to \{\bot,\top\}$ with the following meaning. For $\rho\in \ef{A}{B}$, $D_{(A,B)}[\rho]$ tells, whether there exists a mapping $\overline{\rho}$ with following properties:
\begin{itemize}
\item for every $v\in V(H)$, $\overline{\rho}(v)=\rho(v)$ if $\rho(v)\in (A\cap B)\cup \{\UU\}$ and $\overline{\rho}(v)\in A\setminus B$ if $\rho(v)=\FF$;
\item for every $i=1,2,\ldots,t$, $\overline{\rho}(v_i)=w_i$ if $w_i\in A$ and $\overline{\rho}(v_i)=\UU$ otherwise;
\item for every $e=(v,w)\in E(H)$, $\overline{\rho}(e)$ is a correct path trace with signature $\rho(e)$, beginning in $\overline{\rho}(v)$ if $\overline{\rho}(v)\in A$ and anywhere in $A$ otherwise, ending in $\overline{\rho}(w)$ if $\overline{\rho}(w)\in A$ and anywhere in $A\cap B$ otherwise;
\item path traces $\overline{\rho}(e)$ are edge disjoint.
\end{itemize}
Such mapping $\overline{\rho}$ will be called a {\emph{partial immersion}} of $H$ on $(A,B)$.

For the first separation $(\emptyset,V(T))$ we have exactly one signature with value $\top$, being the signature which maps all the vertices into $\UU$ and all the edges into empty signatures. The result of the whole computation should be the value for the signature for the last separation $(V(T),\emptyset)$, which maps all vertices into $\FF$ and edges into signatures consisting of one pair $(\FF,\FF)$. Therefore, it suffices to show how to fill the values of the table for {\bf{introduce vertex}} step and {\bf{forget vertex}} step.

\medskip
\noindent {\bf{Introduce vertex step.}}

Let us introduce vertex $v\in B\setminus A$ to the separation $(A,B)$, i.e., we consider the new separation $(A\cup\{v\},B)$. Let $\rho\in \ef{A\cup\{v\}}{B}$, we need to show how to compute $D_{(A\cup\{v\},B)}[\rho]$ by careful case study of how the signature $\rho$ interferes with the vertex $v$. If $v=w_i$ for some $i\in \{1,2,\ldots,t\}$, then we consider only such $\rho$ for which $\rho(v_i)=v$; for all the others we fill false values.

\begin{itemize}
\item[{\emph{Case $1$:}}] $v\notin \rho(V(H))$, that is, $v$ is not being mapped on by any vertex of $H$. For every $e\in E(H)$ we consider, how $v$ interferes with $\rho(e)$. We argue that $D_{(A\cup\{v\},B)}[\rho]=\bigvee_{\rho'\in \Gg} D_{(A,B)}[\rho']$ for some set $\Gg$; for every $\rho'\in \Gg$ we require that $\rho'|_{V(H)}=\rho|_{V(H)}$ and for every $e\in E(H)$ we list the possible values of $\rho'(e)$. The set $\Gg$ consists of all the possible signatures on $(A,B)$ that have one of the proper values for every $e\in E(H)$ and also satisfy some additional constraints regarding equivalence relation $\equiv$ that are described further.
\begin{itemize}
\item[{\emph{Case $1.1$:}}] $b_i=v=e_i$ for some pair $(b_i,e_i)\in \rho(e)$. That means, that the signature of $e$ truncated to separation $(A,B)$ must look exactly like $\rho(e)$, but without this subpath of length one. Thus we have one possible value for $\rho'(e)$, being $\rho(e)$ with this pair deleted.
\item[{\emph{Case $1.2$:}}] $b_i=v\neq e_i$ for some pair $(b_i,e_i)\in \rho(e)$ and some $e\in E(H)$. That means that the signature $\rho(e)$ truncated to separation $(A,B)$ has to look the same but for the path corresponding to this very pair, which needs to be truncated by vertex $v$. The new beginning has to be either a vertex in $A\cap B$, or a forgotten vertex from $A\setminus B$. As $T$ is semi-complete and $(A,B)$ is a separation, there is an edge between $v$ and every vertex of $A\setminus B$. Therefore, in $\rho'(e)$ the pair $(b_i,e_i)$ has to be replaced with $(b_i',e_i)$, where $b_i'=\FF$ or $b_i'$ is any vertex of $A\cap B$ such that there is an arc $(v,b_i')$. All the vertices $b_i'\neq \FF$ obtained in this manner have to be pairwise different; moreover, we impose a condition that for all $e\in E(H)$ for which some pair of form $(v,e_i)$ has been truncated to $(\FF,e_i)$, these pairs have to be pairwise non-equivalent with respect to $\equiv$, in order forbid multiple usage of edges going to the forgotten vertices.
\item[{\emph{Case $1.3$:}}] $b_i\neq v=e_i$ for some pair $(b_i,e_i)\in \rho(e)$ and some $e\in E(H)$. Similarly as before, signature $\rho(e)$ truncated to separation $(A,B)$ has to look the same but for the path corresponding to this very pair, which needs to be truncated by vertex $v$. As $(A,B)$ is a separation, the previous vertex on the path has to be in the separator $A\cap B$. Therefore, in $\Gg$ we take into consideration all signatures $\rho'$ such that in $\rho'$ the pair $(b_i,e_i)$ is replaced with $(b_i,e_i')$, where $e_i'$ is any vertex of $A\cap B$ such that there is an arc $(e_i',v)$. Also, all vertices $e_i'$ used in this manner have to be pairwise different.
\item[{\emph{Case $1.4$:}}] $v$ is not contained in any pair in $\rho(e)$. Either $v$ lies in the interior of some subpath from $\overline{\rho}(e)$, or it is not used by $\overline{\rho}(e)$ at all. In the first case the corresponding path in partial immersion on $(A\cup\{v\},B)$ has to be split into two subpaths when truncating the immersion to $(A,B)$. Along this path, the edge that was used to access $v$ had to go from inside $A\cap B$, due to $(A\cup\{v\},B)$ being a separation. However, the edge used to leave $v$ can go to $A\cap B$ as well as to any forgotten vertex from $A\setminus B$, as $(A,B)$ is a separation and $T$ is a semi-complete digraph. In the second case, the signature of the truncated immersion stays the same. Therefore, in $\Gg$ we take into consideration signatures $\rho'$ such that they not differ from $\rho$ on $e$, or in $\rho'(e)$ exactly one pair $(b_i,e_i)$ is replaced with two pairs $(b_i,e_i')$ and $(b_i',e_i)$, where $e_i'\in A\cap B$ with $(e_i',v)\in E(T)$, whereas $b_i'=\FF$ or $b_i'\in A\cap B$ with $(v,b_i')\in E(T)$. Similarly as before, all vertices $b_i'\neq \FF$ used have to be pairwise different and different from those used in case $1.2$, all vertices $e_i'$ used have to be pairwise different and different from those used in case $1.3$, and all the pairs $(\FF,e_i)$ created in this manner have to be pairwise non-equivalent and non-equivalent to those created in case $1.2$.
\end{itemize}
\item[{\emph{Case $2$:}}] $v=\rho(u)$ for some $u\in V(H)$. For every $(u,u')\in E(H)$, $v$ has to be the beginning of the first pair of $\rho((u,u'))$; otherwise, $D_{(A\cup\{v\},B)}=\bot$. Similarly, for every $(u',u)\in E(H)$, $v$ has to be the end of the last pair of $\rho((u',u))$; otherwise, $D_{(A\cup\{v\},B)}=\bot$. Therefore, into $\Gg$ we take all signatures $\rho'$ such that: {\emph{(i)}} $\rho'$ differs on $V(H)$ from $\rho$ only by having $\rho'(u)=v$ where $\rho(u)$ must be equal to $\UU$; {\emph{(ii)}} the first pairs of all $\rho((u,u'))$ are truncated as in the case $1.2$ for all $(u,u')\in E(H)$ (or as in case $1.1$, if the beginning and the end coincide); {\emph{(iii)}} the last pairs of all $\rho((u',u))$ are truncated as in the case $1.3$ for all $(u',u)\in E(H)$ (or as in case $1.1$, if the beginning and the end coincide); {\emph{(iv)}} for all the other edges $e$ we proceed as in the case $1.4$.
\end{itemize}

\medskip\noindent {\bf{Forget vertex step.}}

Let us forget vertex $w\in A\setminus B$ from the separation $(A\cup \{w\},B)$, i.e., we consider the new separation $(A,B)$. Let $\rho\in \ef{A}{B}$; we argue that $D_{(A,B)}[\rho]=\bigvee_{\rho'\in \Gg} D_{(A\cup\{w\},B)}[\rho']$ for some set $\Gg\subseteq \ef{A\cup\{w\}}{B}$, which corresponds to possible extensions of $\rho$ to the previous, bigger separation. We now discuss, which signatures $\rho'$ are needed in $\Gg$ by considering all the signatures $\rho'\in \ef{A\cup\{w\}}{B}$ partitioned with respect to behaviour of the vertex $w$. For every signature $\rho'$ and every $e\in E(H)$ we examine the interaction between $\rho'(e)$ and $w$; in each case we decide, which behaviour of $\rho'(e)$ is correct. In $\Gg$ we take into the account all signatures $\rho'$ behaving properly on all $e\in E(H)$.

\begin{itemize}
\item[{\emph{Case $1$:}}] $w\notin \rho'(V(H))$, that is, $w$ is not an image of a vertex of $H$. In this case we require $\rho'|_{V(H)}=\rho'_{V(H)}$.
\begin{itemize}
\item[{\emph{Case $1.1$:}}] $b_i=w=e_i$ for some pair $(b_i,e_i)\in \rho'(e)$. That means that in the corresponding partial immersions $w$ had to be left to $B\setminus A$; however, in $T$ there is no edge from $w$ to $B\setminus A$ as $(A,B)$ is a separation. Therefore, in $\Gg$ we consider no signatures $\rho'$ behaving in this manner.
\item[{\emph{Case $1.2$:}}] $b_i=w\neq e_i$ for some pair $(b_i,e_i)\in \rho'(e)$. This means that $w$ prolongs some path from the signature $\rho$ in such a manner, that $w$ is its beginning. After forgetting $w$ the beginning of this path belongs to the forgotten vertices; therefore, in $\Gg$ we consider only signatures $\rho'$ in which $\rho'(e)$ differs from $\rho(e)$ on exactly one pair: in $\rho'(e)$ there is $(w,e_i)$ instead of $(\FF,e_i)$ in $\rho$. Moreover, all such pairs $(\FF,e_i)$ that were extended by $w$ have to form one whole equivalence class with respect to $\equiv$ in $\rho$.
\item[{\emph{Case $1.3$:}}] $b_i\neq w=e_i$ for some pair $(b_i,e_i)\in \rho'(e)$. That means that $w$ prolongs some path from the signature $\rho$ in such a manner, that $w$ is its end. As $w\notin\rho'(V(H))$ the path should continue further into $B\setminus A$; however, as $(A,B)$ is a separator, there is no edge between $w$ and $B\setminus A$. We obtain a contradiction; therefore, we take no such signature into account in the set $\Gg$.
\item[{\emph{Case $1.4$:}}] $w$ is not contained in any pair in $\rho'(e)$. In this case $w$ has to be either unused by $\rho'(e)$, or used inside some subpath of $\rho'(e)$. Therefore, we take into consideration in $\Gg$ exactly signatures in which $\rho'(e)=\rho(e)$.
\end{itemize}
\item[{\emph{Case $2$:}}] $w=\rho'(u)$ for some $u\in V(H)$. We consider in $\Gg$ signatures $\rho'$ that differ from $\rho$ in following manner: {\emph{(i)}} $\rho'(u)=w$ where $\rho(u)$ has to be equal to $\FF$; {\emph{(ii)}} for all edges $(u,u')\in E(H)$ the first pair of $\rho'((u,u'))$ is of form $(w,e_1)$, whereas the first pair of $\rho((u,u'))$ is of form $(\FF,e_1)$; {\emph{(iii)}} for all edges $(u',u)\in E(H)$ the last pair of $\rho'((u',u))$ is of form $(b_h,w)$, whereas the last pair of $\rho((u,u'))$ is of form $(b_h,\FF)$; {\emph{(iv)}} for all edges $e\in E(H)$ non-incident with $u$ we follow the same truncation rules as in cases $1.1$-$1.4$. Moreover, all pairs in which $w$ has been replaced with $\FF$ marker have to form one whole equivalence class with respect to $\equiv$.
\end{itemize}

Since updating the table $D_{(A,B)}$ for every separation $(A,B)$ requires at most \linebreak 
$O(|\ef{A}{B}|^2)= 2^{O(k\log p + p\ell(\log \ell+\log p))}$ steps, while the number of separations in the pathwidth decomposition is $O(|V(T)|)$, the theorem follows.
\end{proof}

\section{Counterexample for the irrelevant vertex technique for \\$k$-\textsc{Vertex Disjoint Paths}}\label{sec:app-counter}

We construct a semi-complete digraph $T_n$ with two pairs $(s_1,t_1)$, $(s_2,t_2)$ with the following properties:
\begin{itemize}
\item[(a)] $T_n$ contains an $(n/2-1)$-triple;
\item[(b)] there is exactly one solution to \vertexpath instance $(T_n,\{(s_1,t_1),(s_2,t_2)\})$ in which all the vertices of $V(T_n)$ lie on one of the paths.
\end{itemize}
This example shows that even though a graph can be complicated from the point of view of path decompositions, all the vertices can be relevant.

Let $V(T_n)=\{a_i,b_i: 1\leq i\leq n\}$, where $s_1=a_1$, $t_1=a_n$, $s_2=b_n$ and $t_2=b_1$. Construct following edges:
\begin{itemize}
\item for every $i\in\{1,2,\ldots,n-1\}$ create an edge $(a_i,a_{i+1})$;
\item for every $i,j\in\{1,2,\ldots,n\}, i<j-1$ create an edge $(a_j,a_i)$;
\item for every $i\in\{1,2,\ldots,n-1\}$ create an edge $(b_{i+1},b_i)$;
\item for every $i,j\in\{1,2,\ldots,n\}, i>j+1$ create an edge $(b_j,b_i)$;
\item for every $i\in\{1,2,\ldots,n\}$ create an edge $(a_i,b_i)$;
\item for every $i,j\in\{1,2,\ldots,n\}, i\neq j$ create an edge $(b_j,a_i)$.
\end{itemize}
To see that $T_n$ satisfies (a), observe that 
$$(\{b_1,b_2,\ldots,b_{n/2-1}\},\{a_{n/2+1},a_{n/2+2},\ldots,a_{n}\},\{a_1,a_2,\ldots,a_{n/2-1}\})$$
is a $(n/2-1)$-triple. To prove that (b) is satisfied as well, observe that there is a solution to \vertexpath problem containing two paths $(a_1,a_2,\ldots,a_n)$ and $(b_n,b_{n-1},\ldots,b_1)$ which in total use all the vertices of $T_n$. Assume that there is some other solution and let $k$ be the largest index such that  the path connecting $a_1$ to $a_n$ begins with prefix $(a_1,a_2,\ldots,a_k)$. As the solution is different from the aforementioned one, it follows that $k<n$. Therefore, the next vertex on the path has to be $b_k$, as this is the only unused outneighbor of $a_k$ apart from $a_{k+1}$. But if the path from $a_1$ to $a_n$ already uses $\{a_1,a_2,\ldots,a_k,b_k\}$, we see that there is no edge going from $\{a_{k+1},a_{k+2},\ldots,a_n,b_{k+1},b_{k+2},\ldots,b_n\}$ to $\{b_1,b_2,\ldots,b_{k-1}\}$, so we are already unable to construct a path from $b_n$ to $b_1$. This is a contradiction.

\newpage

\end{document}